\documentclass[a4paper,12pt]{article}

\usepackage{algorithm}
\usepackage{graphicx}
\usepackage{algorithmic}
\usepackage{graphics}
\usepackage{natbib}
\usepackage{amsfonts}
\usepackage{mathrsfs}
\usepackage{amsmath}
\usepackage{amssymb}
\usepackage{amsthm}
\usepackage{bm}
\usepackage{authblk}
\usepackage{epsfig}
\usepackage{cite}
\usepackage{color}
\usepackage[colorlinks,linkcolor=cyan,citecolor=blue,urlcolor=red]{hyperref}

\usepackage[left=2cm,top=2cm,right=2cm,bottom=2cm]{geometry}
\usepackage{setspace}

\numberwithin{equation}{section}
\def\bea{\begin{eqnarray}}
\def\eea{\end{eqnarray}}
\def\ba{\begin{array}}
\def\ea{\end{array}}
\def\nn{\nonumber}
\newtheorem{theorem}{Theorem}

\newtheorem{lemma}{Lemma}

\newtheorem{definition}{Definition}
\newtheorem{condition}{Condition}

\newcommand{\ga}{\Gamma}

\newcommand{\hc}{\hat{c}}
\newcommand{\Ps}{\hat{\Psi}}
\newcommand{\Ph}{\hat{\Phi}}

\newcommand{\N}{\mathcal{N}}
\newcommand{\E}{\mathbb{E}}

\begin{document}
\title{$\frac{\rho}{1-\epsilon}$-approximate pure Nash equilibria algorithms for weighted congestion games and their runtimes}

\author[1]{Chunying Ren\thanks{E-mail: renchunying@emails.bjut.edu.cn}}
\author[2]{Zijun Wu\thanks{E-mail: wuzj@hfuu.edu.cn. Corresponding author}}
\author[1]{Dachuan Xu\thanks{E-mail: xudc@bjut.edu.cn}}
\author[3]{Xiaoguang Yang\thanks{E-mail: xgyang@iss.ac.cn}}

\affil[1]{Department of Operations Research and Information Engineering, Beijing University of Technology, Pingleyuan 100, Beijing,
	100124, China}
\affil[2]{Department of Artificial Intelligence and Big Data, Hefei University, Jinxiu 99, Hefei, 230601, Anhui, China}
\affil[3]{Academy of Mathematics and System Science, University of Chinese Academy of Sciences, 100190, Beijing, China}

%\institute{Department of Operations Research and Scientific Computing, Beijing University of Technology, Beijing 100124, P.R. China. \and
%Faculty of Management, University of New Brunswick, Fredericton, NB Canada E3B 5A3. \and
%School of Mathematics and Statistics, Shandong Normal University, Jinan 250014, P.R. China.}

	\maketitle
	
	\begin{abstract}
This paper concerns computing approximate pure Nash equilibria in weighted congestion games, which has been shown to be PLS-complete. With the help
		of $\Ps$-game and approximate potential functions, we propose two algorithms
	based on best response dynamics, and prove that they efficiently compute $\frac{\rho}{1-\epsilon}$-approximate pure Nash equilibria for $\rho=d!$ and $\rho=\frac{2\cdot W\cdot (d+1)}{2\cdot W+d+1}\le d+1,$
respectively, when the weighted congestion game has polynomial latency functions of degree at most $d\ge1$ and players' weights are bounded from above by a constant $W\ge1$. This improves the recent work of \citet{Feldotto2017} and \citet{Giannakopoulos2022} that showed efficient algorithms for computing $d^{d+o(d)}$-approximate pure Nash equilibria.

\begin{flushleft}
\end{flushleft}
	\end{abstract}
	
	\newpage
	\tableofcontents
	\newpage
	
	\section{Introduction}
	\label{sec:Introduction}
As an important class of non-cooperative games, \emph{congestion games}
	(\citet{Dafermos1969,Rosenthal1973} and \citet{Rougharden2016})
	have been well-studied in the field of \emph{algorithmic game theory}
	(\citet{Roughgarden2010}), in order to understand strategic behavior of selfish players competing over sets of common resources. Much of the work has been devoted to explore properties of \emph{pure Nash equilibria}
	in arbitrary (weighted) congestion games starting from the seminal work of \citet{Rosenthal1973}, which proved
	the existence of pure Nash equilibria in arbitrary unweighted congestion games.
	This mainly includes the \emph{existence} and \emph{inefficiency} of pure Nash equilibria, see, e.g., \citet{Christodoulou2005,Harks2011,Harks2012,Wu2022MP}, and others.

Only few have concerned the problem of finding an equilibrium state in an arbitrary congestion game.
	For the special setting of unweighted congestion games,
	\emph{potential functions} (\citet{Monderer1996}) exist, see, e.g.,
	\citet{Rosenthal1973}, and thus there are relatively simple algorithms driven by \emph{best response dynamics} for computing precise and approximate pure Nash equilibria (see, e.g., \citet{Chien2011} or \citet{Rougharden2016}), although
	\citet{Fotakis2005} have proved that
	finding a pure Nash equilibrium in this case is essentially \emph{PLS-complete} (\citet{Johnson1988}).
	
	For the general setting of weighted congestion games, the situation becomes much worse, as
	pure Nash equilibria then need not exist, see, e.g., \citet{Harks2012}.
	Moreover,
	 \citet{Skopalik2008} have shown that
	computing an \emph{approximate
	pure Nash equilibrium} in this case has already been PLS-complete.
	Hence, we need a more sophisticated algorithm for computing an
	approximate pure Nash equilibrium in this general case. This defines the problem considered in the present paper.
	
	Due to the PLS-completeness, we will not discuss an arbitrary weighted congestion game, but
	these with polynomial latency functions of degree at most $d$ for an arbitrary constant integer $d\ge 1,$
	and with players' weights valued in a bounded interval $[1,W]$ for an arbitrary constant $W\ge 1.$
	Coupling with novel techniques of \emph{$\Ps$-game} by
	\citet{Caragiannis2015} and of \emph{approximate potential functions}
	by \citet{Chen2009}, we propose two simply algorithms based on best response dynamics, which
	compute $\frac{\rho}{1-\epsilon}$-approximate pure Nash equilibria (for $\rho=d!$ and $\rho=\frac{2\cdot W\cdot(d+1)}{2\cdot W+d+1}$, respectively) within polynomial runtimes parameterized by $d$, $A$ and $W$, where $A$ is the largest ratio between two positive coefficients of these polynomial latency functions.
\subsection{Our contribution}
Our first algorithm is essentially a \emph{$\frac{1}{1-\epsilon}$ best response dynamic}
on a $\Ps$-game of a weighted congestion game with polynomial latency functions
for an arbitrary constant $\epsilon\in (0,1)$, see
Algorithm~\ref{alg:Best Response Dynamics}.
\citet{Caragiannis2015} have shown that this $\Ps$-game has a potential function
$\Ph(\cdot)$, and its $\rho$-approximate pure Nash equilibria are essentially
$d!\cdot \rho$-approximate pure Nash equilibria of the original weighted congestion game for each
constant $\rho\ge 1.$

When players share the same strategy set, we prove with similar arguments to these of \citet{Chien2011}
that Algorithm~\ref{alg:Best Response Dynamics} outputs a
$\frac{1}{1-\epsilon}$-approximate pure Nash equilibrium of the $\Ps$-game
within $O(\frac{N\cdot(W+d!\cdot W^{d+1})}{\epsilon}\cdot\log (N\cdot\hc_{\max}))$
iterations, where
$N$ is the number of players and $\hc_{\max}=\max_{S\in\prod_{v\in\N}\Sigma_{v}}
	\max_{u\in\N}\ \hc_u(S)$ is the players' maximum cost value.
%, and
%$\lambda=d!\cdot [(1+\gamma)\cdot W^{d+1}+(1+\frac{1}{\gamma})^d\cdot d^d\cdot W]$
%for an arbitrary constant $\gamma>0$
See Theorem~\ref{th:fast convergence} for a detailed result.

When players have different strategy sets, Algorithm~\ref{alg:Best Response Dynamics}
produces a
$\frac{1}{1-\epsilon}$-approximate pure Nash equilibrium of $\Ps$-game
within $O(\frac{N\cdot\mu}{\epsilon}\cdot\log(N\cdot \hc_{\max}))$
iterations, where $\mu= E\cdot A\cdot(1+d)!\cdot N^d\cdot W^{d+1},$ and $E$ is the number of resources,
see Theorem~\ref{th:a d!-pne of multi-O/D}.

Theorem~\ref{th:fast convergence} and Theorem~\ref{th:a d!-pne of multi-O/D}
together imply that Algorithm~\ref{alg:Best Response Dynamics} computes
a $\frac{d!}{1-\epsilon}$-approximate pure Nash equilibrium of weighted congestion games
with polynomial latency functions within an acceptable runtime. In particular,
these runtimes are polynomial parameterized by $d$, $A$ and $W$.
This improves the recent work of \citet{Feldotto2017} and
\citet{Giannakopoulos2022}, which compute
$d^{d+o(d)}$-approximate pure Nash equilibria with more sophisticated algorithms.

We further improve the result by incorporating the idea of approximate potential function
of \citet{Christodoulou2019}
into our framework, and propose a refined $\frac{\rho}{1-\epsilon}$ best response dynamic for $\rho=\frac{2\cdot W\cdot(d+1) }{2\cdot W+d+1},$
%in particular, its boundary $\frac{2\cdot(1+d)\cdot\lambda\cdot W }{2\cdot\lambda\cdot W+d+1}$ goes from $\frac{2\cdot(1+d)\cdot W }{2\cdot W+d+1}$ (for $\lambda=1$) up to $1+d$ (for $\lambda\rightarrow\infty$).
see
Algorithm~\ref{alg:Best Response Dynamics of appro-PNE}.
Here, we note that \citet{Christodoulou2019} have proved the existence of $\rho$-approximate pure Nash equilibria for each $\rho\in[\frac{2\cdot W\cdot(d+1) }{2\cdot W+d+1},d+1]$.

When players share the same strategy set and $\rho=\frac{2\cdot W\cdot(d+1)}{2\cdot W +d+1}$, Algorithm~\ref{alg:Best Response Dynamics of appro-PNE} computes a $\frac{\rho}{1-\epsilon}$-approximate pure Nash equilibria within $O(\frac{N\cdot (1+W)^{d+1}}{\epsilon}\cdot\log(N\cdot c_{\max}))$ iterations, where $\epsilon\in(0,1)$ is an arbitrary constant, and $c_{\max}$ is the players' maximum cost value in the game $\Gamma$.
See Theorem~\ref{th:fast convergence of appro} for a detailed result. When players have different strategy sets, Algorithm~\ref{alg:Best Response Dynamics of appro-PNE}
produces a
$\frac{\rho}{1-\epsilon}$-approximate pure Nash equilibrium
within $O(\frac{N\cdot \varpi}{\epsilon}\cdot \log(N\cdot c_{\max}))$
iterations, where $\varpi=E\cdot A\cdot(1+d)\cdot N^d\cdot W^{d+1},$
see Theorem~\ref{th:fast convergence of appro for general case}.

While weighted congestion games with polynomial latency functions need not have exact potential functions, our results show that the naive idea of best response dynamics still lead to efficient algorithms for the computation of approximate pure Nash equilibria. In particular, the resulting algorithms compute approximate pure Nash equilibria of much higher accuracy than these proposed recently by \citet{Caragiannis2015}, \citet{Feldotto2017}
and \citet{Giannakopoulos2022}, though our algorithms may have longer polynomial runtimes depending on more parameters.

\subsection{Related work}
\label{subsec:RelatedWorks}
\citet{Rosenthal1973} proved that every unweighted congestion game has
a potential function, and thus admits at least one pure Nash equilibrium.
For weighted congestion games, \citet{Fotakis2005} and \citet{Panagopoulou2007}
showed the existence of potential functions for affine linear and exponential latency functions, respectively.
Hence, such weighted congestion games have pure Nash equilibria.
Beyond these cases, pure Nash equilibria may not exist, see \citet{Goemans2005} and \citet{Harks2012}.
In fact, \citet{Dunkel2008} even proved that deciding if an arbitrary weighted congestion game
has a pure Nash equilibrium is a NP-hard problem.

While each unweighted congestion game has a pure Nash equilibrium,
\citet{Fabrikant2004} have proved that computing such an equilibrium state is  of \emph{PLS-complete} (\citet{Johnson1988}), which is a complexity class locating in between the two
well-known classes P and NP. Moreover, \citet{Ackermannl2008} even showed that
finding a pure Nash equilibrium in a unweighted congestion game with affine linear latency functions
is PLS-complete. This becomes more even tough for weighted congestion games, for which
\citet{Skopalik2008} have proved that finding a $\rho$-approximate pure Nash equilibrium
with a constant $\rho>1$ has already been PLS-complete.

Due to this extreme complexity, we may have to consider particular cases when
we compute approximate pure Nash equilibria in weighted congestion games.
For an arbitrary constant $\epsilon\in (0,1),$
\citet{Chien2011} have shown a $\frac{1}{1-\epsilon}$ best response dynamic, which computes a
$\frac{1}{1-\epsilon}$-approximate pure Nash equilibrium within $O(\frac{N\cdot \alpha}{\epsilon}\log(\frac{\Phi(S^{(0)})}{\Phi_{\min}}))$ iterations when the weighted congestion
game is symmetric and the latency functions fulfill a so-called \emph{$\alpha$ bounded jump condition}, where $\Phi(\cdot)$ is the potential function of the game and $\Phi_{\min}$ is the minimum values of the potential function.

For an arbitrary weighted congestion game $\Gamma$  with polynomial latency functions
of degree at  most an integer $d\ge 1,$
\citet{Caragiannis2015} defined a \emph{$\Ps$-game} with so-called
\emph{$\Ps$-functions}. This $\Ps$-game is essentially a potential game, and shares all
other features with $\Gamma$, but has latency functions dominating
the polynomial latency functions of $\Gamma$. With the help of this $\Ps$-game,
\citet{Caragiannis2015} then showed an algorithm for computing
a $d^{O(d^2)}$-approximate pure Nash equilibrium of $\Gamma$ within $O(g\cdot(1+\frac{m}{\gamma}))$ iterations
with an arbitrary small $\gamma> 0$, where  $m=\text{log}\frac{\hat{c}_{\text{max}}}{\hat{c}_{\text{min}}}$, $\hat{c}_{\text{max}}$ and $\hat{c}_{\min}$
are the respective maximum and minimum cost of players of the $\Ps$-game, $g=(1+m\cdot(1+\gamma^{-1}))^d\cdot d^d\cdot N^2\cdot\gamma^{-4}$, and $N$ is the number of players. Coupling with certain random mechanisms, \citet{Feldotto2017} improved this result of \citet{Caragiannis2015}, and showed
a randomized algorithm that efficiently outputs a
$d^{d+o(d)}$-approximate pure Nash equilibrium of $\Gamma$ with a high
probability.

Recently, \citet{Christodoulou2019} initiated a different research with the notion of an \emph{approximate potential function}, which was first proposed by \citet{Chen2009}. They proved that every weighted congestion
game with polynomial latency functions of degree at most $d\ge 1$ has
a $\rho$-approximate pure Nash equilibrium for each
$\rho \in[\frac{2\cdot W\cdot (d+1)}{2\cdot W+d+1},d+1],$ where
$W$ is the common upper bound of
players' weights.
Inspired by \citet{Christodoulou2019}, \citet{Giannakopoulos2022} showed a deterministic polynomial-time algorithm
for computing $d^{d+o(d)}$-approximate pure Nash equilibrium of weighted congestion games within $O\left(g\cdot(1+(d+1)\cdot p)\cdot(m\cdot p+d+1)\right)$ iterations when the latency functions are polynomial with degree at most $d\ge 1,$
where  $m=\text{log}\frac{c_{\text{max}}}{c_{\text{min}}}$, $c_{\text{max}}$ and $c_{\min}$
are the respective maximum and minimum cost of players, $g=N^2\cdot p^3\cdot(1+m\cdot(1+p))^d\cdot d^d+1$, $p=(2\cdot d+3)\cdot(d+1)\cdot(4\cdot d)^{d+1}=d^{d+o(d)}$, and $N$ is the number of players. To our knowledge, this is the best known result in the literature for the computation of approximate pure Nash equilibria in weighted congestion games.

We generalize the runtime analysis of $\frac{1}{1-\epsilon}$ best response dynamics in \citet{Chien2011} from symmetric weighted congested games with $\alpha$ bounded jump latency functions to arbitrary weighted congestion games with polynomial latency functions. In particular, our results outperform these in the literature from the perspective of obtaining a preciser approximation of pure Nash equilibria.

\subsection{Outline of the  paper}
\label{subsec:Arrangement}
The remaining of the paper is organized as follows. Section~\ref{se:models and preliminaries} defines the weighted congestion game model, and presents preliminary properties of $\Ps-$games and of approximate potential functions. We propose the two algorithms and analyze their runtimes in Section~\ref{sec:Algorithms_and_Analysis}. A short summary is then presented in Section~\ref{se:conclusion}.
	
\section{Model and Preliminaries}\label{se:models and preliminaries}
\subsection{Weighted congestion games}\label{se:models and preliminaries of game}
We represent an arbitrary \emph{weighted congestion game} $\ga$ by
tuple $(\N,\E,(\Sigma_u)_{u\in\N},\\(c_e)_{e\in\E},w)$ with components
defined below.
\begin{description}
		\item[G1.]\label{G1} $\N$ is a collection of $N$ different players.
        \item[G2.]\label{G2} $\E$  is a set of $E\ge 1$ different resources.
        \item[G3.]\label{G3}  $\Sigma_u$ is a set of feasible strategies of player $u\in\N$. Here, each strategy $s_u\in \Sigma_u$  is a \emph{nonempty} subset of $\E$, that is, $s_u\in \Sigma_u\subseteq
         2^{\E}\setminus\{\emptyset\}=\{B\ne \emptyset:\ B\subseteq \E\}$.
        \item[G4.]\label{G4} $w=(w_u)_{u\in\N}$ is a nonempty and finite vector, where each component
		$w_u>0$ is a positive \emph{weight} of player $u\in \N$ that denotes
		an \emph{unsplittable demand} and
		will be delivered by a single strategy $s_u\in\Sigma_u$.
		When $w_u=w_{u'}$ for all $u,u'\in \N,$ then $\ga$ is called \emph{unweighted}.
		Otherwise, $\Gamma$ is called \emph{weighted}.
 		\item[G5.]\label{G5}
 		$c=(c_e)_{e\in \E}$ is a vector of \emph{latency (or cost) functions} of resources $e\in \E$.
		Here, each latency function $c_e: [0,\infty)\to [0,\infty)$ is
		continuous, nonnegative and nondecreasing.
	\end{description}

Players are noncooperative, and each of them chooses a strategy independently.
This then results in a \emph{(pure strategy) profile} $S=(s_u)_{u\in\N}$, where
$s_u$ denotes the strategy used by player $u\in\N.$ For each resource
$e\in \E,$ we denote by $U_e(S):=\{w_u:\ e\in s_u\ \forall u\in\N\}$
a \emph{multi-set} consisting of weights of players $u\in \N$ using
resource $e$ in profile $S.$ Moreover, we denote by
$L(M)$ the \emph{sum} $\sum_{m\in M} m$ of all elements in a multi-set $M.$ Then
$L(U_e(S))$ is the total weight of players using resource $e$ in profile
$S.$ Furthermore, resource $e\in \E$ has a \emph{latency} of
$c_e(L(U_e(S)))$, and  player $u\in\N$ has a \emph{cost} of
$c_u(S)=w_u\cdot \sum_{e\in s_u} c_e(L(U_e(S)))$,
when player $u$ uses a strategy $s_u\in\Sigma_u$ in profile $S.$
Finally, the profile $S$ has a \emph{(total) cost} of
$C(S):=\sum_{u\in\N} c_u(S)=\sum_{e\in\E} L(U_e(S))\cdot c_e(L(U_e(S))).$

To facilitate our discussion, for each player $u\in\N,$ we denote by $S_{-u}$ a \emph{subprofile}
$(s_{u'})_{u'\in \N\setminus\{u\}}$ of strategies used by
``opponents'' $u'\in\N\setminus\{u\}$ of player $u,$ and write
$S=(s_u,S_{-u})$ when we need to show explicitly the
strategy $s_u$ used by player $u.$

We call a profile $S$ a \emph{pure Nash equilibrium}
if
\begin{equation}\label{def:PNE}
	c_u(S)=c_u(s_u,S_{-u})\le c_u(s_u',S_{-u})
\end{equation}
for each player $u\in\N$ and each strategy $s_u'\in\Sigma_u$.
Here, we recall that $c_u(S)$ is the cost of player $u$ in profile $S,$ and
that $c_u(s'_u,S_{-u})$ is the resulting cost of player $u$ when player $u$ unilaterally moves from
strategy $s_u$ to another strategy $s'_u.$ Inequality~\eqref{def:PNE} then simply states that
$s_u$ is a \emph{best-response} to $S_{-u}$ for each player $u\in \N$   when $S=(s_u,S_{-u})$ is a
pure Nash equilibrium. Hence, there is no incentive for a player to unilaterally change
strategy when the profile is already a pure Nash equilibrium.

\citet{Monderer1996} have shown that pure Nash equilibria exist in an arbitrary finite game with
a \emph{potential function}, i.e., a \emph{potential game}. For our model, a real-valued function
$\Phi:\prod_{u\in\N}\Sigma_u\to [0,\infty)$ is called an (exact) \emph{potential function}
if
\bea\label{def:Exact_Potentials} \Phi(s'_u,S_{-u})\!-\!\Phi(s_u,S_{-u})\!&&=\!c_u(s'_u,S_{-u})\!-\!c_u(s_u,S_{-u}),
\nn\\
&&\quad \forall
	u\!\in\!\N\ \forall s_u,s'_u\!\in\!\Sigma_u\ \forall S_{-u}\!\in\! \prod_{u'\!\in\!\N\!\setminus\!\{u\}}\!\Sigma_{u'}.
\eea
Hence, a potential function quantifies the cost reduction of a unilateral change of strategy.
In particular, its global minimizers are pure Nash equilibria.

\citet{Rosenthal1973} has shown that an arbitrary unweighted congestion game
has a potential function, and so admits pure Nash equilibria. For weighted congestion games,
\citet{Panagopoulou2007} have proved the existence of potential functions for exponential
latency functions, and \citet{Fotakis2005} have proved the existence of potential functions for
affine linear latency functions.
Beyond these cases, potential functions need not exist, and neither need pure Nash equilibria,
see, e.g., \citet{Harks2012}. Then we have to consider a relaxed notion of \emph{$\rho$-approximate
pure Nash equilibrium} instead.

Formally, for an arbitrary constant $\rho\ge 1$, a profile $S$ is called \emph{$\rho$-approximate pure Nash equilibrium}
if
\begin{equation}\label{def:rho_PNE}
	c_u(S)=c_u(s_u,S_{-u})\le \rho\cdot c_u(s'_u,S_{-u})
\end{equation}
for each player $u\in\N$ and each strategy $s_u'\in\Sigma_u$.
Inequality~\eqref{def:rho_PNE} means that a unilateral change of strategy reduces
the cost at a rate of at most $\frac{c_u(S)-c_u(s'_u,S_{-u})}{c_u(S)}\le 1-\frac{1}{\rho}$ in a $\rho$-approximate pure Nash equilibrium.
Hence, the smaller the constant $\rho$ is, the closer the profile $S$ would approximate a precise
pure Nash equilibrium. To facilitate our discussion, we call $\rho$ the approximation ratio of $S$ when $S$ is a $\rho$-approximate pure Nash equilibrium.

As \citet{Skopalik2008} have shown that computing an approximate pure Nash equilibrium in an arbitrary congestion game is PLS-complete, we thus consider such a computation
in a parametric setting. We focus only on weighted congestion games with \emph{polynomial}
latency functions fulfilling Conditions~\ref{con:restrictions on weights}--\ref{con:restrictions on latency functions} below.
\begin{condition}\label{con:restrictions on weights}
$w_u\in[1,W]$ for each player $u\in\N$ for a constant $W\ge1$.
\end{condition}
\begin{condition}\label{con:restrictions on latency functions}
Each latency function $c_e: [0,\infty)\to [0,\infty)$ is polynomial, and has a form of
	\begin{equation}\label{eq:latency_functions}
		c_e(x)=\sum_{k=0}^d a_{e,k}\cdot x^k\quad \forall x\in [0,\infty)\  \forall e\in \E,
	\end{equation}
where $d\ge 1$ is an integer, $a_{e,k}\ge 0$ and
$\sum_{k'=0}^d a_{e,k'}>0$ for all $k=0,1,\ldots,d$ and all $e\in \E.$
In particular, there is a constant $A>0$ such that
$\frac{a_{e',k'}}{a_{e,k}}\le A$ for arbitrary $e,e'\in \E$ and arbitrary
$k,k'\in \{0,1,\ldots,d\}$ with $a_{e,k}>0.$
\end{condition}

Condition~\ref{con:restrictions on weights} simply states that each player has a weight
in	$[1,W]$ for a constant upper bound $W\ge 1.$ Condition~\ref{con:restrictions on latency functions}
requires that each of the polynomial latency functions has a degree at most $d\ge 1,$ and
the ratios between their positive coefficients are bounded from above by a constant $A>0.$
Such weighted congestion games need not have precise pure Nash equilibria, see, e.g., \citet{Harks2012}.

In Section~\ref{sec:Algorithms_and_Analysis}, we will propose two algorithms computing
$\frac{\rho}{1-\epsilon}$-approximate pure Nash equilibria for
$\rho=d!$ and $\frac{2\cdot W\cdot(d+1)}{2\cdot W+d+1},$ respectively, and prove that their runtimes
are polynomial parameterized by $d,$ $A$ and $W.$
The first algorithm is essentially a \emph{$\frac{1}{1-\epsilon}$ best response dynamic} on
\emph{$\Ps$-game}, while the second algorithm is a \emph{refined $\frac{\rho}{1-\epsilon}$
best response dynamic} for $\rho=\frac{2\cdot W\cdot(d+1)}{2\cdot W+d+1}.$
Their runtime analyses will involve properties of $\Ps$-games and of \emph{approximate potential functions}. Before we formally define these two algorithms, let us first introduce the respective notions of
$\hat{\Psi}$-games and of approximate potential functions in Section~\ref{subsec:Ps_game} and Section~\ref{subsec:Approximate_Potential_Function}, so as to
facilitate our further discussion.
\subsection{$\Ps-$games}
\label{subsec:Ps_game}
While weighted congestion games fulfilling Conditions~\ref{con:restrictions on weights}--\ref{con:restrictions on latency functions} need not have
potential functions, \citet{Caragiannis2015}
have shown that a suitable revision of the latency functions with the so-called \emph{$\Ps$-functions} would result
in a \emph{$\Ps$-game}, which has an exact potential function.

\begin{definition}[$\Ps$-functions, see \citet{Caragiannis2015}]\label{de:Ps-function}
Consider an arbitrary integer $k\ge 1$ and a finite nonempty ground set $X$
of reals.	A \emph{($k$-order) $\Ps$-function} on $X$ is a real-valued function $\Ps_k:\mathcal{M}(X)\to [0,\infty)$ with
$$
\Ps_k(M)=k!\cdot
\sum_{(r_x)_{x\in M}\in \mathbb{N}_{\ge 0}^{M}:\ \sum_{x\in M} r_x=k}\ \ \prod_{x\in M} x^{r_x}
$$
for each multi-set $M\in\mathcal{M}(X),$ where $\mathcal{M}(X)$ is the collection of
all multi-sets with elements from $X.$ Here, we put $\Ps_k(\emptyset)=0$ for each integer $k\ge 1.$ Moreover,
we employ a convention that $\Ps_{0}(M)=1$ for an arbitrary multi-set $M$.
\end{definition}

Clearly, $\Ps_k(M)$ coincides with a summation of all monomials of (total) degree $k$ over elements in $M$. In particular, $\Ps_1(M)=L(M)$. Moreover. $\Ps_k(M)$ and $L(M)^k$ share the same monomial terms, though different coefficients.
Lemma~\ref{le:the properties of the Ph-function} below collects some trivial properties
of these $\Ps$-functions. Readers may refer to \citet{Caragiannis2009} for their proofs.

\begin{lemma}[Properties of $\Ps$-functions, see \citet{Caragiannis2015}]\label{le:the properties of the Ph-function}
	Consider an arbitrary integer $k\ge 1$, an arbitrary finite multi-set $M$ of
	nonnegative reals, and an arbitrary nonnegative constant real $b.$
	Then the following four statements hold.
	\begin{itemize}
		\item[a.] $L(M)^k\le\Ps_{k}(M)\le k!\cdot L(M)^k.$
		\item[b.] $\Ps_k(M)\le k\cdot \Ps_1(M)\cdot \Ps_{k-1}(M).$
		\item[c.] $\Ps_k(M\cup\{b\})\le(\Ps_k(\{b\})^{1/k}+\Ps_k(M)^{1/k})^k.$
        \item[d.] $\Ps_k(M\cup\{b\})-\Ps_k(M)=k\cdot b\cdot \Ps_{k-1}(M\cup\{b\}).$
	\end{itemize}
\end{lemma}

With these $\Ps-$functions, we are now ready to  introduce the notion of
\emph{$\Ps$-games} proposed by \citet{Caragiannis2015}.

\begin{definition}[$\Ps$-games, see also \citet{Caragiannis2015}]\label{Ps-game}
	Consider a weighted congestion game
	$\Gamma=(\N,\E,(\Sigma_u)_{u\in\N},(c_e)_{e\in\E},w)$ fulfilling Conditions~\ref{con:restrictions on weights}--\ref{con:restrictions on latency functions}.
	The \emph{$\Ps$-game} of $\Gamma$ refers to another weighted congestion game
	$\hat{\Gamma}=(\N,\E,(\Sigma_u)_{u\in\N},(\hat{c}_e)_{e\in\E},w),$ which
	shares the same components with $\Gamma$, but different latency functions
	$\hat{c}_e: [0,\infty)\to [0,\infty)$ satisfying equality~\eqref{def:Ps_Latency_Function} below,
\begin{equation}\label{def:Ps_Latency_Function}
	\hat{c}_e(L(U_e(S)))=\sum_{k=0}^d a_{e.k}\cdot \Ps_{k}(U_e(S))\quad\forall
	S\in \prod_{u\in\N}\Sigma_u\ \forall e\in \E.
\end{equation}
\end{definition}

To simplify notation, we write
$\hc_e(L(U_e(S)))$ simply as $\hc_e(S)$ when we discuss the latency of
a resource $e$ w.r.t. a profile $S$ in $\Ps$-game $\hat{\Gamma}.$
Moreover, we employ $\hat{c}_u(S)=w_u\cdot \sum_{e\in s_u}\hat{c}_e(S)
=w_u\cdot \sum_{e\in s_u}\sum_{k=0}^d a_{e,k}\cdot \Ps_{k}(U_e(S))$ and
$\hat{C}(S)=\sum_{u\in\N} \hat{c}_u(S)$ to denote the respective cost of a player $u\in\N$ and
of a profile $S$ in $\Ps$-game $\hat{\Gamma}.$
Clearly, $\Ps$-game $\hat{\Gamma}$ would coincide with $\Gamma,$ and so
$\hat{c}_u(S)=c_u(S)$ and $\hat{C}(S)=C(S),$ when $d=1.$
In general, they are different. Nonetheless,
Lemma~\ref{le:the properties of the Ph-function}a yields immediately that
$\hat{c}_u(S)\in [c_u(S),\ d!\cdot c_u(S)]$ for each $u\in\N$ and each
$S\in \prod_{u'\in\N}\Sigma_{u'}.$ We summarize this in Lemma~\ref{le:the relationship of a profile} below.

\begin{lemma}[\citet{Caragiannis2015}]\label{le:the relationship of a profile}
	Consider a weighted congestion game $\Gamma$ fulfilling Conditions~\ref{con:restrictions on weights}--\ref{con:restrictions on latency functions}, and its $\Ps$-game
	$\hat{\Gamma}$. Then we obtain for an arbitrary player $u\in\N$ and for an arbitrary
	profile $S$ that
	$
	c_u(S)\le \hat{c}_u(S)\le d!\cdot c_u(S).
	$
\end{lemma}

While Lemma~\ref{le:the relationship of a profile} is trivial, it implies that
a profile $S$ is a $\frac{d!}{1-\epsilon}$-approximate pure Nash equilibrium of $\Gamma$
if $S$ is a $\frac{1}{1-\epsilon}$-approximate pure Nash equilibrium of $\Ps$-game $\hat{\Gamma}$ for each constant $\epsilon\in (0,1).$
This follows immediately from inequality~\eqref{def:rho_PNE}.
We summarize it in Lemma~\ref{le:the relation of equilibrium} below.
\begin{lemma}[\citet{Caragiannis2015}]\label{le:the relation of equilibrium}
	Consider an arbitrary weighted congestion game $\Gamma$ fulfilling Conditions~\ref{con:restrictions on weights}--\ref{con:restrictions on latency functions}, an arbitrary constant $\epsilon\in (0,1),$ and an arbitrary profile
	$S\in \prod_{u\in\N}\Sigma_{u}$ of $\Gamma.$
	If $S$ is a $\frac{1}{1-\epsilon}$-approximate pure Nash equilibrium of $\Ps$-game $\hat{\Gamma}$ of
	$\Gamma,$ then $S$ is a $\frac{d!}{1-\epsilon}$-approximate pure Nash equilibrium of $\Gamma.$
\end{lemma}

Lemma~\ref{le:the relation of equilibrium} indicates that we can obtain a $\frac{d!}{1-\epsilon}$-approximate
pure Nash equilibrium of a weighted congestion game $\Gamma$ by computing a $\frac{1}{1-\epsilon}$-approximate pure Nash equilibrium of the $\Ps$-game
$\hat{\Gamma}$ when $\Gamma$ fulfills Conditions~\ref{con:restrictions on weights}--\ref{con:restrictions on latency functions}
and when $\hat{\Gamma}$ has a $\frac{1}{1-\epsilon}$-approximate pure Nash equilibrium.
Theorem~\ref{th:potential function of the ps-game} below confirms this idea by showing that
$\Ps$-game $\hat{\Gamma}$ has a potential function, and so admits a precise pure Nash equilibrium when the
weighted congestion game $\Gamma$ fulfills Conditions~\ref{con:restrictions on weights}--\ref{con:restrictions on latency functions}. Readers may refer to \citet{Caragiannis2015} for a proof.

\begin{theorem}[Existence of potential functions in $\Ps$-games, see \citet{Caragiannis2015}]\label{th:potential function of the ps-game}
Consider a weighted congestion game $\Gamma$ fulfilling Conditions~\ref{con:restrictions on weights}--\ref{con:restrictions on latency functions}. Then its $\Ps$-game $\hat{\Gamma}$ has
a potential function $\Ph:\prod_{u\in\N}\Sigma_{u}\to [0,\infty)$ with
$$
\Ph(S)=\sum_{e\in \E}\sum_{k=0}^d\frac{a_{e,k}}{k+1}\cdot \Ps_{k+1}(U_e(S))
$$
for each profile $S\in \prod_{u\in\N}\Sigma_{u}.$
Moreover, $\Gamma$ has $d!$-approximate pure Nash equilibria, which are pure Nash equilibria
of $\hat{\Gamma}.$
\end{theorem}

With Theorem~\ref{th:potential function of the ps-game}, we will employ a $\frac{1}{1-\epsilon}$ best response dynamic on the $\Ps$-game $\hat{\Gamma}$ to compute a $\frac{d!}{1-\epsilon}$-approximate
pure Nash equilibrium of $\Gamma.$ This then results in Algorithm~\ref{alg:Best Response Dynamics} in
Section~\ref{subsec:Ps_game_best_response_dynamic}.

Lemma~\ref{le:lower-upper of potential} below bounds the potential function
$\Ph(S)$ of an arbitrary profile $S$ with the total cost $\hat{C}(S).$
While this is trivial, it will be very helpful when we derive the runtime of Algorithm~\ref{alg:Best Response Dynamics} in Section~\ref{subsec:Analysis_Best_Response}.
We omit its proof due to its triviality.

\begin{lemma}[\citet{Caragiannis2015}]\label{le:lower-upper of potential}
	Consider an arbitrary weighted congestion game $\Gamma$ fulfilling Conditions~\ref{con:restrictions on weights}--\ref{con:restrictions on latency functions}.
	Let $\hat{\Gamma}$ be its $\Ps$-game defined in Definition~\ref{Ps-game}, and let
	$\Ph(\cdot)$ be the potential function of $\hat{\Gamma}$ as given in Theorem~\ref{th:potential function of the ps-game}. Then,
for every profile $S$,
\[
1\le\Ph(S)\le\hat{C}(S)=\sum_{u\in\N}\hc_u(S).
\]
\end{lemma}
Lemma~\ref{le:lower-upper of potential} yields immediately that
	the potential function value $\Ph(S)$ is bounded from above by
\bea\label{eq:hc_upper_bound}
	N\cdot\max_{u\in\N} \hc_u(S)&\le& N\cdot W\cdot E\cdot  \max_{e\in\E} \hat{c}_e(S)
	\le N\cdot W\cdot E\cdot d!\cdot \max_{e\in\E} c_e(N\cdot W)\nn\\
	&=&N\cdot W\cdot E\cdot d!\cdot \max_{e\in\E}
	\sum_{k=0}^d a_{e,k}\cdot N^k\cdot W^k\nn\\
    &\le&
	N^{d+1}\cdot W^{d+1}\cdot E\cdot (d+1)!\cdot \max_{e\in \E}\max_{k=0}^d a_{e,k},
\eea
which implies that the runtime of Algorithm~\ref{alg:Best Response Dynamics} is polynomial parameterized by the three constants $d,$ $A$ and $W$.
We will come to this later in Section~\ref{subsec:Analysis_Best_Response}.

\subsection{Approximate potential functions}
\label{subsec:Approximate_Potential_Function}

In addition to above idea of $\frac{1}{1-\epsilon}$ best response dynamic on
$\Ps(\cdot),$ we will also consider a refined best response dynamic with the idea of
\emph{approximate potential function} proposed by \citet{Chen2009} in Section~\ref{se:a approximate potential func}.

\begin{definition}[Approximate potential functions, see \citet{Chen2009}]
	\label{def:approximate_potential}
	Consider a weighted congestion game $\Gamma$ fulfilling Conditions~\ref{con:restrictions on weights}--\ref{con:restrictions on latency functions}, and consider
	a constant $\rho\ge 1.$ We call
	$\Phi:\prod_u\Sigma_u\to [0,\infty)$ a $\rho$-approximate potential function of $\Gamma$ if
	$$
	\Phi(S)-\Phi(s_{u}',S_{-u})\ge c_u(S)-\rho\cdot  c_u(s_{u}',S_{-u}),\quad
	\forall S\in \prod_{u\in\N}\Sigma_{u}\ \forall u\in \N\ \forall s'_u\in \Sigma_{u}.
	$$
\end{definition}

With this notion,
	\citet{Christodoulou2019} have shown the existence of $\rho$-approximate pure Nash equilibria
	in weighted congestion games fulfilling Conditions~\ref{con:restrictions on weights}--\ref{con:restrictions on latency functions} for each constant
$\rho\in[\frac{2\cdot W\cdot (d+1)}{2\cdot W+d+1},d+1].$
In particular, they showed that
\begin{equation}\label{eq:approximate_potential_function}
	\Phi(S)=\sum_{e\in \E} \Psi_e(L(U_e(S))), \quad \forall S\in\prod_{u\in\N}\Sigma_u
\end{equation}
defines a $\rho$-approximate potential function for $\rho=\frac{2\cdot W\cdot (d+1)}{2\cdot W+d+1}$
when %$\Psi_e: [0,\infty)\to [0,\infty)$ with
$$
\Psi_e(x):=a_{e,0}\cdot x+\sum_{k=1}^da_{e,k}\cdot \left(\frac{x^{k+1}}{k+1}+\frac{x^k}{2}\right),\quad \forall x\ge 0\ \forall e\in \E.
$$

With this approximate potential function, we design a refined best response dynamic
 for computing $\frac{2\cdot W\cdot (d+1)}{(2\cdot W+d+1)\cdot (1-\epsilon)}$-approximate pure Nash equilibria, see Algorithm~\ref{alg:Best Response Dynamics of appro-PNE} in Section~\ref{subsubsec:Refined_Best_Reponse_Dynamic}. Again, to facilitate
the resulting runtime analysis, we introduce a trivial upper bound of this approximate potential
function in Lemma~\ref{le:lower-upper approx potential} below.

\begin{lemma}\label{le:lower-upper approx potential}
	Consider an arbitrary weighted congestion game $\Gamma$ fulfilling Conditions~\ref{con:restrictions on weights}--\ref{con:restrictions on latency functions}.
	Then, for each profile $S$, $\Phi(S)\le C(S)=\sum_{u\in\N}c_{u}(S)$.
\end{lemma}
\begin{proof}
Lemma~\ref{le:lower-upper approx potential} follows since
\bea
\Phi(S)&=&\sum_{e\in\E}\Psi_e(L(U_e(S)))\nn\\
&=&\sum_{e\in\E}\left[
a_{e,0}\cdot L(U_e(S))+\sum_{k=1}^da_{e,k}\cdot
\left[\frac{L(U_e(S))^{k+1}}{k+1}+\frac{L(U_e(S))^k}{2}\right]
\right]\nn\\
&\le&\sum_{e\in\E}\left[
a_{e,0}\cdot L(U_e(S))+\sum_{k=1}^da_{e,k}\cdot
L(U_e(S))^{k+1}
\right]\nn\\
&\le&\sum_{e\in\E}L(U_e(S))\cdot \left[
\sum_{k=0}^da_{e,k}\cdot
L(U_e(S))^{k}
\right]=C(S)\nn
\eea
for an arbitrary profile $S\in \prod_{u\in\N}\Sigma_{u}.$
Here, we used the definition~\eqref{eq:approximate_potential_function} and that $\frac{L(U_e(S))^k}{2}\le\frac{k\cdot L(U_e(S))^k}{k+1}
\le\frac{k\cdot L(U_e(S))^{k+1}}{k+1}$ for $k\ge1$.
\end{proof}

Similar to the potential function $\Ph(\cdot)$ of the $\Ps$-game,
Lemma~\ref{le:lower-upper approx potential} implies that the approximate potential function
$\Phi(\cdot)$
has an upper bound of
\begin{equation}\label{eq:c_upper_bound}
N\cdot \max_{u\in \N} c_u(S)\le N^{d+1}\cdot W^{d+1}\cdot E\cdot (d+1)
\cdot \max_{e\in\E}\max_{k=0}^d\ a_{e,k},
\end{equation}
which implies that Algorithm~\ref{alg:Best Response Dynamics of appro-PNE} has a polynomial runtime
parameterized by the three constants $d,$ $A$ and $W$,
see Section~\ref{subsubsec:Runtime_Analysis_Refined_Best_Reponse_Dynamic} below.

\section{Computing $\frac{\rho}{1-\epsilon}$-approximate pure Nash equilibria}
\label{sec:Algorithms_and_Analysis}

We now propose two algorithms based on best response dynamics for computing
	$\frac{\rho}{1-\epsilon}$-approximate pure Nash equilibria when the constant
	$\rho$ equals $d!$ and $\frac{2\cdot W\cdot (d+1)}{2\cdot W+d+1},$ respectively.
	For the case of $\rho=d!,$ we apply a $\frac{1}{1-\epsilon}$ best response dynamic on the
	$\Ps$-game similar to that in
	\citet{Chien2011}. This results in our Algorithm~\ref{alg:Best Response Dynamics}.
	For the case of $\rho=\frac{2\cdot W\cdot (d+1)}{2\cdot W+d+1},$ we
	design a refined best response dynamic by incorporating the idea of approximate potential function,
which forms our Algorithm~\ref{alg:Best Response Dynamics of appro-PNE}.

\subsection{A $\frac{1}{1-\epsilon}$ best response dynamic on $\Ps$-game}
\label{se:Ps-game}
\subsubsection{The algorithm}
\label{subsec:Ps_game_best_response_dynamic}

Consider an arbitrary profile $S\in \prod_{u\in\N}\Sigma_{u},$ an arbitrary constant $\rho>1$ and an arbitrary player $u\in\N.$
We call a strategy $s'_u\in \Sigma_{u}$ a \emph{$\rho$-move} of player $u$
in $\Ps$-game $\hat{\Gamma}$
if
\begin{equation}\label{eq:def_epsilon_move}
	\rho\cdot \hat{c}_u(s'_u,S_{-u})<\hat{c}_u(s_u,S_{-u})=\hat{c}_u(S).
\end{equation}

Similarly, one may define $\rho$-moves directly for game $\Gamma.$
Inequality~\eqref{eq:def_epsilon_move} actually means that player $u$ can reduce cost at a rate
of
$\frac{\hc_{u}(S)-\hc_{u}(s'_u,S_{-u})}{\hc_{u}(S)}\ge 1-\frac{1}{\rho}$ by unilaterally moving to strategy $s'_u$ when $s'_u$ is a $\rho$-move.
Moreover, when no player has a $\rho$-move, then the profile
has been a $\rho$-approximate pure Nash equilibrium.

Algorithm~\ref{alg:Best Response Dynamics} below shows a \emph{$\frac{1}{1-\epsilon}$ best response
dynamic} of $\Ps$-game of a weighted congestion game fulfilling
Conditions~\ref{con:restrictions on weights}--\ref{con:restrictions on latency functions} for an arbitrary constant $\epsilon\in (0,1)$, which shares similar features with
the dynamic proposed in \citet{Chien2011} for \emph{symmetric} congestion games
with $\alpha$ bounded jump latency functions. It starts with an arbitrary initial profile
$S^{(0)}=(s_{u}^{(0)})_{u\in\N}\in \prod_{u\in\N}\Sigma_{u},$ and then  evolves the profile
by iterating the following three steps over the time horizon $t\in\mathbb{N}$ until a
$\frac{1}{1-\epsilon}$-approximate pure Nash equilibrium of $\Ps$-game is met.
\begin{algorithm}[!h]
	 \caption{ A $\frac{1}{1-\epsilon}$ best response dynamic of $\Ps$-game}
    \label{alg:Best Response Dynamics}
    \begin{algorithmic}[1]
    	\REQUIRE A $\Ps$-game and a constant $\epsilon\in (0,1)$
    	\ENSURE $\frac{1}{1-\epsilon}$-approximate pure Nash equilibrium
    	\STATE choose an arbitrary initial profile $S^{(0)},$ and put $t=0$
    	\WHILE{$S^{(t)}$ is not a $\frac{1}{1-\epsilon}$-approximate pure Nash equilibrium of  $\Ps$-game}
    	\FOR{each $u\in \N$}
    	\STATE compute $s_u^*=\arg\min_{s'_u\in \Sigma_{u}} \hat{c}(s'_u,S^{(t)}_u)$
    	\ENDFOR
    	\STATE $\N_t^*=\{u\in\N:\ \hat{c}_u(S^{(t)})>\frac{1}{1-\epsilon}\cdot \hat{c}_u(s_u^*,S_{-u}^{(t)}) \}$
    	\STATE pick an arbitrary player $u_t$ from $\N_t^*$ fulfilling condition that
    	\bea\label{eq:Best_response_Optimal_Player}
    	\hat{c}_{u_t}(S^{(t)})-\hat{c}_{u_t}(s^*_{u_t},S_{-u_t}^{(t)})
    		\ge \hat{c}_{u}(S^{(t)})-\hat{c}_{u}(s^*_u,S^{(t)}_{-u})\quad \forall u\in\N_t^*
    	\eea
       \STATE $S^{(t+1)}=(s^*_{u_t},S^{(t)}_{-u_t})$ and $t=t+1$
    	\ENDWHILE
    	\RETURN $S^{(t)}$
    \end{algorithmic}
\end{algorithm}
\begin{description}
	\item[Step 1.] When current profile
	$S^{(t)}$ is not a $\frac{1}{1-\epsilon}$-approximate pure Nash equilibrium of
	$\Ps$-game $\hat{\Gamma},$ then computes for each $u\in\N$ a \emph{best-response}\footnote{The best response $s_u^*$ depends essentially on $S_{-u}^{(t)}$, and is thus a function
		of $S^{(t)}_{-u}.$ Nevertheless, we denote it
simply by $s_u^*$, so as to simplify notation.}
$s_u^*=
	\arg\min_{s'_u\in \Sigma_{u}}\hat{c}_u(s'_u,S^{(t)}_{-u})$ w.r.t. subprofile $S^{(t)}_{-u},$ and
	puts all players with $\frac{1}{1-\epsilon}$-moves into a collection $\N_t^*,$ i.e.,
	\[
	\forall u\in\N:\ u\in\N^*_t \iff \hat{c}_u(S^{(t)})=\hc_{u}(s_u^{(t)},S_{-u}^{(t)})>\frac{1}{1-\epsilon}\cdot \hat{c}_u(s_u^*,S_{-u}^{(t)}).
	\]
	Here, we note that $s_u^{(t)}$ is the strategy of player $u\in\N$ in profile $S^{(t)}$, and that $\N^*_t$ is not empty when current profile $S^{(t)}$ is not a $\frac{1}{1-\epsilon}$-approximate
	pure Nash equilibrium.
	\item[Step 2.] Choose an arbitrary player $u_t$ from $\N^*_t$ fulfilling condition that
	\begin{equation}\label{eq:Choice_of_Player_in_iteration_t}
		\hat{c}_{u_t}(S^{(t)})-\hat{c}_{u_t}(s_{u_t}^*,S_{-u_t}^{(t)})\ge \hat{c}_{u}(S^{(t)})-\hat{c}_{u}(s_{u}^*,S_{-u}^{(t)})\quad \forall u\in\N^*_t.
	\end{equation}
  \item[Step 3.] Let the selected player $u_t$ move
  unilaterally from $s_{u_t}^{(t)}$ to $s_{u_t}^*,$ and then put $S^{(t+1)}=(s_{u_t}^*,S_{-u_t}^{(t)}).$
\end{description}

We assume w.l.o.g. that each iteration of Algorithm~\ref{alg:Best Response Dynamics} has a polynomial complexity. This is true when the congestion game is defined on a graph, i.e., when it is a network congestion game, for which the best response of a player is computed efficiently by a polynomial time shortest path algorithm, e.g., Dijkstra's algorithm in~\citet{Dijkstra1959}. Then the runtime complexity of Algorithm~\ref{alg:Best Response Dynamics} depends on the total number of iterations essentially.

Define $T_\epsilon(S^{(0)}):=\arg\min\ \{t\in\mathbb{N}:\ S^{(t)}\text{ is a }\frac{1}{1-\epsilon}
\text{-approximate pure Nash} \text{equilibrium of }\hat{\Gamma}\}$ and define
$T_\epsilon:=\max_{S^{(0)}\in \prod_{u\in\N}\Sigma_u}T_\epsilon(S^{(0)})$.
Then $T_\epsilon(S^{(0)})$ is the \emph{runtimes} of Algorithm~\ref{alg:Best Response Dynamics}
w.r.t. initial profile $S^{(0)}$ (i.e., the number of iterations Algorithm~\ref{alg:Best Response Dynamics}
takes for finding a $\frac{1}{1-\epsilon}$-approximate pure Nash equilibrium
of $\hat{\Gamma}$ when the initial profile is $S^{(0)}$), and
$T_\epsilon$ is the corresponding maximum runtime.
Section~\ref{subsec:Analysis_Best_Response} below inspects the upper bound
of $T_\epsilon$ with the potential function
$\Ph(\cdot)$ defined in Theorem~\ref{th:potential function of the ps-game}.

\subsubsection{Runtime analysis of Algorithm~\ref{alg:Best Response Dynamics}}
\label{subsec:Analysis_Best_Response}

Note that the cost
\begin{equation}\label{eq:Natural_Player_Cost_Lower_Bound}
\hat{c}_u(S)=w_u\cdot \sum_{e\in s_u}\sum_{k=0}^d a_{e,k}\cdot\Ps_k(U_e(S))\ge
a_{\min}^+>0\quad\forall u\in\N\ \forall S\in \prod_{v\in\N}\Sigma_v,
\end{equation}
where $a_{\min}^+:=\min_{e\in\E} \min\{a_{e,k}:\ k=0,\ldots,d\text{ with }a_{e,k}>0\}$ is
the \emph{minimum positive coefficient}.
This inequality follows since each player has a weight not smaller than $1$
(Condition~\ref{con:restrictions on weights}), and since the latency functions fulfill Condition~\ref{con:restrictions on latency functions}.

Note also that for each $t<T_\epsilon(S^{(0)}),$
\begin{equation}\label{eq:Potential_Decremental_Best_Response}	\Ph(S^{(0)})\ge\Ph(S^{(t)})-\Ph(S^{(t+1)})=\hat{c}_{u_t}(S^{(t)})-\hat{c}_{u_t}(S^{(t+1)})>
	\epsilon\cdot \hat{c}_{u_t}(S^{(t)}).
\end{equation}
This follows since $s_{u_t}^*$ is essentially a $\frac{1}{1-\epsilon}$-move of player $u_t,$
and since $\Ph(\cdot)$ is a potential function of $\Ps$-game $\hat{\Gamma},$
see Theorem~\ref{th:potential function of the ps-game}.

Inequality~\eqref{eq:Potential_Decremental_Best_Response} together with inequality~\eqref{eq:Natural_Player_Cost_Lower_Bound} yield immediately that
\[
\Ph(S^{(t)})-\Ph(S^{(t+1)})\ge \epsilon\cdot a_{\min}^+>0
\]
for each iteration $t$ with $t<T_\epsilon(S^{(0)}).$
Moreover, we obtain
that $T_\epsilon(S^{(0)})\le \frac{\Ph(S^{(0)})-\Ph_{\min}}{\epsilon\cdot a_{\min}^+}$ for an arbitrary
$\epsilon\in (0,1)$ and an arbitrary initial profile $S^{(0)},$
where $\Ph_{\min}$ is  \emph{minimum} potential values.

This together with Lemma~\ref{le:lower-upper of potential} further imply that
$T_\epsilon\le \frac{N\cdot\hc_{\max}}{\epsilon\cdot a_{\min}^+}<\infty,$ where
$\hc_{\max}$ is short for
$\max_{S\in \prod_{v\in \N}\Sigma_v}\max_{u\in\N}\hc_u(S).$
Hence, Algorithm~\ref{alg:Best Response Dynamics} terminates within finite iterations.
We summarize this Theorem~\ref{theo:Direct_Running_Time_Best_Response_Dynamic} below.

\begin{theorem}\label{theo:Direct_Running_Time_Best_Response_Dynamic}
	Let $\Gamma$ be a weighted congestion game fulfilling Conditions~\ref{con:restrictions on weights}--\ref{con:restrictions on latency functions}, let $\hat{\Gamma}$ be its $\Ps$-game
	as in Definition~\ref{Ps-game}, let $S^{(0)}$ be an arbitrary initial profile, and let $\epsilon\in (0,1)$
	be an arbitrary constant.
	Then Algorithm~\ref{alg:Best Response Dynamics} computes a $\frac{1}{1-\epsilon}$-approximate
	pure Nash equilibrium of $\Ps$-game $\hat{\Gamma}$ within
	$T_{\epsilon}(S^{(0)})\le \frac{\Ph(S^{(0)})-\Ph_{\min}}{\epsilon\cdot a_{\min}^+}\le \frac{N\cdot\hc_{\max}}{\epsilon\cdot a_{\min}^+}$ iterations.
	Moreover, $T_\epsilon\le \frac{N\cdot\hc_{\max}}{\epsilon\cdot a_{\min}^+},$ and so Algorithm~\ref{alg:Best Response Dynamics} outputs a $\frac{d!}{1-\epsilon}$-approximate pure Nash equilibrium of
	$\Gamma$ within $\frac{N\cdot\hc_{\max}}{\epsilon\cdot a_{\min}^+}$ iterations.
\end{theorem}

While Theorem~\ref{theo:Direct_Running_Time_Best_Response_Dynamic}
shows a finite upper bound $\frac{N\cdot\hc_{\max}}{\epsilon\cdot a_{\min}^+}$ of $T_\epsilon,$
it may overestimate $T_\epsilon$, as it is based on a very crude estimation of
player's cost in inequality \eqref{eq:Natural_Player_Cost_Lower_Bound}.
To obtain a tighter upper bound, we implement a finer analysis below.

With Lemma~\ref{le:lower-upper of potential}, we obtain that
\begin{equation}\label{eq:Potential_Upper_Bound}
\hc_{m_t}(S^{(t)})=\max_{u\in\N}\ \hat{c}_u(S)\ge \frac{1}{N}\cdot\hat{C}(S)\ge \frac{1}{N}\cdot \Ph(S).
\end{equation}
Here, $m_t\in\N$ is a player with a maximum cost
$\hc_{m_t}(S^{(t)})=\max_{u\in\N}\hc_u(S^{(t)}).$
When $\hc_{u_t}(S^{(t)})=\hc_{m_t}(S^{(t)})$ in an iteration $t<T_\epsilon(S^{(0)}),$ then
inequalities~\eqref{eq:Potential_Upper_Bound} and \eqref{eq:Potential_Decremental_Best_Response}
together yield that
\begin{equation*}%\label{eq:Potential_Decremental_Special_Best_Response}
\Ph(S^{(t)})-\Ph(S^{(t+1)})\ge \frac{\epsilon}{N}\cdot \Ph(S^{(t)}),
\end{equation*}
and the potential function value decreases at a constant ratio of
at least $\frac{\epsilon}{N}$ in this case. This would yield a tighter upper bound
$O(\frac{N}{\epsilon}\cdot \log (N\cdot\hc_{\max}))$ of $T_\epsilon$ by
a similar proof to that of Theorem~\ref{th:fast convergence} below.

However, in general,  the selected player $u_t$ may have a cost
$\hc_{u_t}(S^{(t)})<\hc_{m_t}(S^{(t)}).$ Then the above analysis does not apply.
Nevertheless, Lemma~\ref{le:The relationship between the cost of O/D pairs} below
shows a similar result that
\begin{equation}\label{eq:Potential_Decremental_Special_Best_Response}
	\Ph(S^{(t)})-\Ph(S^{(t+1)})\ge \frac{\epsilon}{N\cdot(W+d!\cdot W^{d+1})}\cdot \Ph(S^{(t)})
\end{equation}
for each $t< T_\epsilon(S^{(0)})$,
when all players share the same strategy set, i.e.,
$\Sigma_{u}=\Sigma_{v}$ for two arbitrary players
$u,v\in\N.$
We move its proof to
Appendix~\ref{app:Proof_Of_Lemma_relationship_between_cost_of_OD_pairs}.
\begin{lemma}\label{le:The relationship between the cost of O/D pairs}
	Consider an arbitrary weighted congestion game $\Gamma$
	fulfilling Conditions~\ref{con:restrictions on weights}--\ref{con:restrictions on latency functions},
	and its $\Ps$-game $\hat{\Gamma}.$
	If $\Sigma_{u}=\Sigma_{u'}$ for all $u,u'\in \N,$ then
	\bea\label{eq:Cost_Reduction_lowe_bound_dynamic_best_response}
	\hc_{u_t}(S^{(t)})-\hc_{u_t}(S^{(t+1)})\ge\frac{\epsilon}{W+d!\cdot W^{d+1}}\cdot \hc_{u}(S^{(t)})
	\eea
	for every player $u\in\N$ and for each $t<T_\epsilon(S^{(0)})$, where $W$ is the common upper bound of players' weights.
Moreover, inequalities \eqref{eq:Potential_Upper_Bound} and \eqref{eq:Cost_Reduction_lowe_bound_dynamic_best_response} together yield
	\[
	\Ph(S^{(t)})-\Ph(S^{(t+1)})\ge \frac{\epsilon}{N\cdot(W+d!\cdot W^{d+1})}\cdot \Ph(S^{(t)})
	\]
	for each $t<T_\epsilon(S^{(0)})$ and for each initial profile $S^{(0)}$.
\end{lemma}

Lemma~\ref{le:The relationship between the cost of O/D pairs}
implies that Algorithm~\ref{alg:Best Response Dynamics}
computes a $\frac{1}{1-\epsilon}$-approximate pure Nash equilibrium of the
$\Ps$-game within
$O(\frac{N\cdot(W+d!\cdot W^{d+1})}{\epsilon}\cdot
\log(N\cdot\hc_{\max}))$ iterations when $\Sigma_{u}=\Sigma_{u'}$ for
all $u,u'\in\N.$ We summarize this result in
Theorem~\ref{th:fast convergence} below.

\begin{theorem}[Particular case]\label{th:fast convergence}
	Consider a weighted congestion game $\Gamma$ fulfilling Conditions~\ref{con:restrictions on weights}--\ref{con:restrictions on latency functions}, and its $\Ps$-game $\hat{\Gamma}.$
	if $\Sigma_{u}=\Sigma_{u'}$ for all $u,u'\in\N,$ then
	Algorithm~\ref{alg:Best Response Dynamics} computes a $\frac{1}{1-\epsilon}$-approximate
	pure Nash equilibrium of $\hat{\Gamma}$ within
$O(\frac{N\cdot(W+d!\cdot W^{d+1})}{\epsilon}\cdot
\log(N\cdot\hc_{\max}))$
	iterations,
$N=|\N|$ is the number of players, $\epsilon\in (0,1)$ is a small
	constant, and $\hc_{\max}=\max_{S\in\prod_{v\in\N}\Sigma_{v}}\max_{u\in\N}\hc_u(S)$ is the player's maximum cost value.
\end{theorem}

\begin{proof}
	Inequality~\eqref{eq:Potential_Decremental_Special_Best_Response} yields immediately
	that
\bea
	\Ph(S^{(t+1)})&\le& \left(1-\frac{\epsilon}{N\cdot(W+d!\cdot W^{d+1})}\right)\cdot \Ph(S^{(t)})\nn\\
	&\le& \left(1-\frac{\epsilon}{N\cdot(W+d!\cdot W^{d+1})}\right)^{t+1}\cdot \Ph(S^{(0)}).\nn
\eea
	for each $t<T_\epsilon(S^{(0)}).$
Note that $(1-x)^{1/x}\le(e^{-x})^{1/x}=\frac{1}{e}$ for $x\in(0,1)$.
Hence, Algorithm~\ref{alg:Best Response Dynamics} will terminate
 within at most
$\frac{N\cdot(W+d!\cdot W^{d+1})}{\epsilon}\cdot \log \frac{\Ph(S^{(0)})}{\Ph_{\min}}$
iterations, since, otherwise, the potential function value would decrease to a value below its minimum $\Ph_{\min}$. Lemma~\ref{le:lower-upper of potential} yields immediately that
$\frac{N\cdot(W+d!\cdot W^{d+1})}{\epsilon}\cdot \log \frac{\Ph(S^{(0)})}{\Ph_{\min}}\le\frac{N\cdot (W+d!\cdot W^{d+1})}{\epsilon}\cdot \log(N\cdot\hc_{\max})$ for each initially profile $S^{(0)}\in\prod_{u\in\N}\Sigma_u$.
\end{proof}

Theorem \ref{th:fast convergence} and Lemma~\ref{le:the relationship of a profile} together
imply that Algorithm~\ref{alg:Best Response Dynamics} produces a
$\frac{d!}{1-\epsilon}$-approximate pure Nash equilibrium of $\Gamma$ within
$O(\frac{N\cdot(W+d!\cdot W^{d+1})}{\epsilon}\cdot \log(N\cdot\hc_{\max}))$ iterations in this particular case. With inequality~\eqref{eq:hc_upper_bound}, we
can see that this runtime is polynomial in input size when
it is parameterized by $d$ and $W.$ Here, we note that
$\epsilon$ is an arbitrary small constant.

In the above,
inequality~\eqref{eq:Cost_Reduction_lowe_bound_dynamic_best_response} plays a crucial
role for bounding $T_\epsilon(S^{(0)})$.
While this need not hold in general, Lemma~\ref{le:The relationship between the cost of multi-O/D} below shows a similar
result when there are players $u,u'\in\N$ with $\Sigma_{u}\ne \Sigma_{u'}$.
We move its proof to Appendix~\ref{proof:The relationship between the cost of multi-O/D}.

\begin{lemma}\label{le:The relationship between the cost of multi-O/D}
	Consider an weighted congestion game $\Gamma$ fulfilling Conditions~\ref{con:restrictions on weights}--\ref{con:restrictions on latency functions}, and its $\Ps$-game $\hat{\Gamma}.$
	Then for each $t<T_\epsilon(S^{(0)}),$ we have
\bea\label{eq:The relationship between the cost of O/D pairs} \hc_{u_t}(S^{(t)})-\hc_{u_t}(s_{u_t}^{*},S^{(t)}_{-u_t})\ge\frac{\epsilon}{\mu}\cdot \hc_{u}(S^{(t)})
	\eea
for every player $u\in\N$ and a constant $\mu:=E\cdot A\cdot(1+d)!\cdot N^d\cdot W^{d+1}$,
	where $E=|\E|$ is the number of resources, $A=\max\{\frac{a_{e,k}}{a_{e',k'}}:
	\ e,e'\in\E,\text{ } k,k'=0,\ldots,d\text{ with }a_{e',k'}>0\}$
	is the maximum ratio between two positive coefficients, and
	$W$ is the common upper bound of players' weights.
\end{lemma}

Lemma~\ref{le:The relationship between the cost of multi-O/D} yields immediately that
Algorithm~\ref{alg:Best Response Dynamics} computes a $\frac{1}{1-\epsilon}$-approximate
pure Nash equilibrium of $\Ps$-game within $O(\frac{ N \cdot\mu }{\epsilon}\cdot \log(N\cdot\hc_{\max}))$ iterations. We summarize this in
Theorem~\ref{th:a d!-pne of multi-O/D} below.

\begin{theorem}[General case]\label{th:a d!-pne of multi-O/D}
Consider an arbitrary weighted congestion game $\Gamma$ fulfilling Conditions~\ref{con:restrictions on weights}--\ref{con:restrictions on latency functions}, and its $\Ps$-game $\hat{\Gamma}$.
Let $\epsilon\in (0,1)$ be an arbitrary constant.
Then $T_\epsilon\le \frac{N\cdot\mu }{\epsilon}\cdot\log(N\cdot\hc_{\max}).$
where $\mu$ is a constant defined in Lemma~\ref{le:The relationship between the cost of multi-O/D}.
Moreover, Algorithm~\ref{alg:Best Response Dynamics} computes a
$\frac{d!}{1-\epsilon}$-approximate pure Nash equilibrium of
$\Gamma$ within
$O(\frac{N\cdot\mu}{\epsilon}\cdot\log(N\cdot\hc_{\max}))$ iterations.
\end{theorem}

While the runtime in Theorem~\ref{th:a d!-pne of multi-O/D}
	is larger than that in Theorem~\ref{th:fast convergence}, it is still polynomial parameterized
	by the constants $d$, $A$ and $W$. This means that Algorithm~\ref{alg:Best Response Dynamics}
computes a $\frac{d!}{1-\epsilon}$-approximate pure Nash equilibrium of an arbitrary weighted
congestion game fulfilling Conditions~\ref{con:restrictions on weights}--\ref{con:restrictions on latency functions} within a polynomial runtime regardless whether
$\Sigma_{u}=\Sigma_{u'}$ for all $u',u\in\N$ or not.
This generalizes the result of \citet{Chien2011} that shows a similar result only for
symmetric congestion games with $\alpha$ bounded jump latency functions.
In particular, Algorithm~\ref{alg:Best Response Dynamics} has a better approximation ratio of $\frac{d!}{1-\epsilon}$ than those in the recent seminal work of \citet{Caragiannis2015}, \citet{Feldotto2017}
and \citet{Giannakopoulos2022}, although its runtime might be longer.
Section~\ref{se:a approximate potential func} below will propose a refined best response dynamic,
which computes a ``more precise'' approximate pure Nash equilibrium even within a much shorter runtime.

\subsection{A refined $\frac{2\cdot W\cdot (d+1)}{(2\cdot W+d+1)\cdot (1-\epsilon)}$ best response dynamic}\label{se:a approximate potential func}

Algorithm~\ref{alg:Best Response Dynamics} implements a
$\frac{1}{1-\epsilon}$ best response dynamic on the $\Ps$-game of a weighted congestion game
$\Gamma$ fulfilling Conditions~\ref{con:restrictions on weights}--\ref{con:restrictions on latency functions}.
With the approximate potential function $\Phi(\cdot)$ in Section~\ref{subsec:Approximate_Potential_Function},
we now define a refined best response dynamic directly on $\Gamma,$
and show that the resulting algorithm efficiently computes a $\frac{2\cdot W\cdot (d+1)}{(2\cdot W+d+1)
\cdot (1-\epsilon)}$-approximate pure Nash equilibrium of $\Gamma$ with a shorter runtime than Algorithm~\ref{alg:Best Response Dynamics}.

\subsubsection{The algorithm}
\label{subsubsec:Refined_Best_Reponse_Dynamic}

Algorithm~\ref{alg:Best Response Dynamics of appro-PNE} below shows
the pseudo code of a refined
 $\frac{\rho}{1-\epsilon}$ best response dynamic of $\Gamma$ for
 $\rho=\frac{2\cdot W\cdot (d+1)}{2\cdot W+d+1}$
 and an arbitrary small constant $\epsilon\in(0,1).$
It shares same structures with Algorithm~\ref{alg:Best Response Dynamics}, except for
	the choice of the player $u_t$.
Algorithm~\ref{alg:Best Response Dynamics} picks a player $u_t=\arg\max_{u\in\N_t^*}
\{c_u(S^{(t)})-c_u(s_u^*,S_{-u}^{(t)})\},$ while
Algorithm~\ref{alg:Best Response Dynamics of appro-PNE} chooses
a player $u_t=\arg\max_{u\in\N_t^*}
\{c_u(S^{(t)})-\rho\cdot c_u(s_u^*,S_{-u}^{(t)})\}$. This also distinguishes
Algorithm~\ref{alg:Best Response Dynamics of appro-PNE} from these commonly used best response dynamics,
see, e.g., \citet{Nisan2007}. An advantage of letting such a player to update strategy is that the approximate potential function $\Phi(\cdot)$ (defined in Section~\ref{subsec:Approximate_Potential_Function}) then decreases very rapidly, since
\bea\label{eq:Approximate_Potential_Decrement_Lower_Upper_General} \Phi(S^{(t)})-\Phi(S^{(t+1)})&=&\Phi(S^{(t)})-\Phi(s_{u_t}^*,S_{-u_t}^{(t)})\ge c_{u_t}(S^{(t)})-\rho\cdot c_{u_t}(s_{u_t}^*,S_{-u_t}^{(t)})\nn\\
		&\ge& c_u(S^{(t)})-\rho\cdot c_u(s_u^*,S_{-u}^{(t)})\ge
		 \epsilon\cdot c_{u}(S^{(t)})
\eea
for all $u\in\N_t^*$ when $S^{(t)}$ is not a $\frac{\rho}{1-\epsilon}$-approximate pure Nash equilibrium. This will then facilitate the runtime analysis in Section~\ref{subsubsec:Runtime_Analysis_Refined_Best_Reponse_Dynamic}.

\begin{algorithm}[!h]
	 \caption{ A refined $\frac{\rho}{1-\epsilon}$ best response dynamic of $\Gamma$}
    \label{alg:Best Response Dynamics of appro-PNE}
    \begin{algorithmic}[1]
   	\REQUIRE A weighted congestion game $\Gamma$ fulfilling Conditions~\ref{con:restrictions on weights}--\ref{con:restrictions on latency functions},  $\rho=\frac{2\cdot W\cdot (d+1)}{2\cdot W+d+1}$, and an arbitrary small constant $\epsilon\in (0,1)$
   	\ENSURE A $\frac{\rho}{1-\epsilon}$-approximate pure Nash equilibrium of $\Gamma$
   	\STATE take an arbitrary initial profile $S^{(0)}$ and put $t=0$
   	\WHILE{$S^{(t)}$ is not a $\frac{\rho}{1-\epsilon}$-approximate pure Nash equilibrium}
   	\FOR{each $u\in\N$}
   	\STATE $s_u^*=\arg\min_{s_u\in \Sigma_{u}} c_u(s_u, S_{-u}^{(t)})$
   	\ENDFOR
   	\STATE $\N_t^*=\{u\in\N_t:\ u\text{ has }\frac{\rho}{1-\epsilon}\text{-moves}\}$, i.e.,
   	\[
   	\forall u\in\N: u\in\N_t^* \iff c_u(S^{(t)})>\frac{\rho}{1-\epsilon}\cdot c_u(s_u^*,S_{-u}^{(t)})
   	\]
   	\STATE $u_t:=\arg\max_{u\in\N_t^*} \{c_u(S^{(t)})-\rho\cdot c_u(s_u^*,S_{-u}^{(t)})\}$
   	\STATE $S^{(t+1)}=(s_{u_t}^*,S_{-u_t}^{(t)})$ and $t=t+1$
   	\ENDWHILE
   	\RETURN $S^{(t)}$
   \end{algorithmic}
\end{algorithm}
Similar to Algorithm~\ref{alg:Best Response Dynamics}, we assume that each iteration of
Algorithm~\ref{alg:Best Response Dynamics of appro-PNE} has a polynomial complexity.
Then the runtime (i.e., the total number of iterations Algorithm~\ref{alg:Best Response Dynamics of appro-PNE} takes for computing a $\frac{\rho}{1-\epsilon}$-approximate pure Nash equilibrium)
dominates the total computational complexity. To facilitate our discussion, we employ
$T_{\rho,\epsilon}(S^{(0)})$ to denote the runtime of Algorithm~\ref{alg:Best Response Dynamics of appro-PNE} when the initial profile is $S^{(0)}.$ Then
$T_{\rho,\epsilon}=\max_{S^{(0)}\in \prod_{u\in\N}\Sigma_u}T_{\rho,\epsilon}(S^{(0)})$
is the worst-case runtime of Algorithm~\ref{alg:Best Response Dynamics of appro-PNE}.

\subsubsection{Runtime analysis of Algorithm~\ref{alg:Best Response Dynamics of appro-PNE}}
\label{subsubsec:Runtime_Analysis_Refined_Best_Reponse_Dynamic}

Theorem~\ref{th:fast convergence of appro} below shows that
Algorithm~\ref{alg:Best Response Dynamics of appro-PNE} terminates
within $O(\frac{N\cdot (1+W)^{d+1}}{\epsilon}\cdot
\log (N\cdot c_{\max}))$
iterations, for the constant $\rho=\frac{2\cdot W\cdot (d+1)}{2\cdot W+d+1}$
and for each constant $\epsilon\in (0,1)$ when $\Sigma_{u}=\Sigma_{u'}$ for two arbitrary players
$u,u'\in\N.$ Here, $c_{\max}=\max_{u\in\N} \max_{S\in \prod_{v\in \N}\Sigma_v} c_u(S)$
is the maximum cost of a player in game $\Gamma$.

\begin{theorem}[Particular case]\label{th:fast convergence of appro}
	Consider a weighted congestion game $\Gamma$ fulfilling Conditions~\ref{con:restrictions on weights}--\ref{con:restrictions on latency functions},
	a constant $\rho=\frac{2\cdot W\cdot (d+1)}{2\cdot W+d+1},$
	a constant $\epsilon\in (0,1).$
	If $\Sigma_{u}=\Sigma_{u'}$ for two arbitrary players
	$u,u'\in\N,$ then Algorithm~\ref{alg:Best Response Dynamics of appro-PNE}
	computes a $\frac{\rho}{1-\epsilon}$-approximate pure Nash equilibrium of
	$\Gamma$ within
$O(\frac{N\cdot (1+W)^{d+1}}{\epsilon}\cdot
\log (N\cdot c_{\max}))$
iterations, where $c_{\max}=\max\{c_u(S):\forall S\in\prod_{u\in\N}\Sigma_u\text{ and }\forall u\in\N\}$ is the maximum cost of a player.
\end{theorem}

The runtime in Theorem~\ref{th:fast convergence of appro} is polynomial parameterized by
	$d$ and $W,$ since $\log c_{\max}$ is essentially a polynomial in input size, see inequality \eqref{eq:c_upper_bound}. Hence, Algorithm~\ref{alg:Best Response Dynamics of appro-PNE}
efficiently computes a $\frac{\rho}{1-\epsilon}$-approximate pure Nash equilibrium when
$\Sigma_{u}=\Sigma_{u'}$ for all $u,u'\in\N.$
This improves the result of Theorem~\ref{th:fast convergence}.

We prove Theorem~\ref{th:fast convergence of appro} below with a similar argument
to that of Theorem~\ref{th:fast convergence}.

Similar to Lemma~\ref{le:The relationship between the cost of O/D pairs},
Lemma~\ref{le:The relationship between the cost of O/D pairs with appro}
below shows that $\Phi(S^{(t)})-\Phi(S^{(t+1)})$ is bounded from below by
$\frac{\epsilon}{(1+W)^{d+1}}\cdot \max_{u\in\N} c_u(S^{(t)})$ for an arbitrary
$t<T_{\rho,\epsilon}(S^{(0)})$ when $\Sigma_{u}=\Sigma_{u'}$ for
two arbitrary players $u,u'\in\N.$
This combining with Lemma~\ref{le:lower-upper approx potential}
imply that the approximate potential value decreases at a rate of
at least $\frac{\epsilon}{N\cdot (1+W)^{d+1}}$ in each iteration before Algorithm~\ref{alg:Best Response Dynamics of appro-PNE} terminates. Moreover, Theorem~\ref{th:fast convergence of appro}
follows immediately from an argument similar to that of
Theorem~\ref{th:fast convergence}.

\begin{lemma}\label{le:The relationship between the cost of O/D pairs with appro}
	Consider a weighted congestion game $\Gamma$ fulfilling Conditions~\ref{con:restrictions on weights}--\ref{con:restrictions on latency functions},
	a constant $\rho=\frac{2\cdot W\cdot (d+1)}{2\cdot W+d+1},$ and
	a constant $\epsilon\in (0,1).$
	If $\Sigma_{u}=\Sigma_{u'}$ for all $u,u'\in\N,$ then
	%When $t<T_{\rho(\lambda),\epsilon}(S^{(0)}),$ then
	\[
	\Phi(S^{(t)})-\Phi(S^{(t+1)})\ge c_{u_t}(S^{(t)})-\rho\cdot c_{u_t}(S^{(t+1)})\ge \frac{\epsilon}{(1+W)^{d+1}}\cdot
	c_u(S^{(t)})
	\]
	for each player $u\in\N,$ and for each $t<T_{\rho,\epsilon}(S^{(0)}).$
\end{lemma}

We move the proof of Lemma~\ref{le:The relationship between the cost of O/D pairs with appro} to Appendix~\ref{app:The relationship between the cost of O/D pairs with appro}.
Lemma~\ref{le:The relationship between the cost of O/D pairs with appro} plays a crucial role for bounding $T_{\rho,\epsilon}(S^{(0)})$ in the above particular case. Similarly, we can generalize it, see Lemma~\ref{le:The relationship between the cost of multi-O/D with appro} below.

\begin{lemma}\label{le:The relationship between the cost of multi-O/D with appro}
	Consider a weighted congestion game $\Gamma$ fulfilling Conditions~\ref{con:restrictions on weights}--\ref{con:restrictions on latency functions},
	a constant $\rho=\frac{2\cdot W\cdot (d+1)}{2\cdot W+d+1},$ and
	a constant $\epsilon\in (0,1).$
	Then for each $t<T_{\rho,\epsilon}(S^{(0)}),$ we have
		\[
	\Phi(S^{(t)})-\Phi(S^{(t+1)})\ge c_{u_t}(S^{(t)})-\rho\cdot c_{u_t}(S^{(t+1)})\ge \frac{\epsilon}{\varpi}\cdot \hc_{u}(S^{(t)})
	\]
	for every player $u\in\N$ and a constant $\varpi:=E\cdot A\cdot(1+d)\cdot N^{d}\cdot W^{d+1}$,
	where $E$ is the number of resources, $W$ is the common upper bound of players' weights, $A=\max\{\frac{a_{e,k}}{a_{e',k'}}:
	\ e,e'\in\E, \text{  } k,k'=0,\ldots,d\text{ with }a_{e',k'}>0\}$
	is the maximum ratio between two positive coefficients, and $N=|\N|$ is the number of players.
\end{lemma}

We move its proof to Appendix~\ref{app:The relationship between the cost of multi-O/D pairs with appro}.

Lemma~\ref{le:The relationship between the cost of multi-O/D with appro} yields immediately that
Algorithm~\ref{alg:Best Response Dynamics of appro-PNE} computes a $\frac{\rho}{1-\epsilon}$-approximate
pure Nash equilibrium of the game $\Gamma$ within $O(\frac{N\cdot\varpi }{\epsilon}\cdot \log
(N\cdot c_{\max}))$ iterations for $\rho=\frac{2\cdot W\cdot (d+1)}{2\cdot W+d+1}$
and for each constant $\epsilon\in (0,1),$ where
$\varpi=E\cdot A\cdot(1+d)\cdot N^{d}\cdot W^{d+1}$. We summarize this in
Theorem~\ref{th:fast convergence of appro for general case} below.

\begin{theorem}[General case]\label{th:fast convergence of appro for general case}
	Consider a weighted congestion game $\Gamma$ fulfilling Conditions~\ref{con:restrictions on weights}--\ref{con:restrictions on latency functions},
	a constant $\rho=\frac{2\cdot W\cdot (d+1)}{2\cdot W+d+1}$,
	a constant $\epsilon\in (0,1).$
Then Algorithm~\ref{alg:Best Response Dynamics of appro-PNE}
	computes a $\frac{\rho}{1-\epsilon}$-approximate pure Nash equilibrium of
	$\Gamma$ within
$O(\frac{N\cdot\varpi}{\epsilon}
\log (N\cdot c_{\max}))$
iterations, where $c_{\max}$ is the maximum cost of players, and
$\varpi=E\cdot A\cdot(1+d)\cdot N^{d}\cdot W^{d+1}$.
\end{theorem}

Theorem~\ref{th:fast convergence of appro for general case} means that
Algorithm~\ref{alg:Best Response Dynamics of appro-PNE} computes a
$\frac{\rho}{1-\epsilon}$-approximate pure Nash equilibrium in a polynomial runtime
for $\rho=\frac{2\cdot W\cdot (d+1)}{2\cdot W+d+1}$ and an arbitrary small constant
$\epsilon\in (0,1),$ when the weighted congestion game fulfills Conditions~\ref{con:restrictions on weights}--\ref{con:restrictions on latency functions}.
Hence, Algorithm~\ref{alg:Best Response Dynamics of appro-PNE} improves Algorithm~\ref{alg:Best Response Dynamics} from both the perspectives of accuracy and of efficiency.
 While the runtime of Algorithm~\ref{alg:Best Response Dynamics of appro-PNE} in Theorem~\ref{th:fast convergence of appro for general case} depends on more parameters than
these in the recent seminal work of \citet{Caragiannis2015,Feldotto2017} and \citet{Giannakopoulos2022},
it computes an approximate pure Nash equilibrium with a much better approximation ratio of $\frac{2\cdot W\cdot (d+1)}{(2\cdot W+d+1)\cdot (1-\epsilon)}.$

\section{Summary}\label{se:conclusion}

We consider the computation of approximate pure Nash equilibria
in weighted congestion games with polynomial latency functions. We
design two algorithms based on best response dynamics.
The first algorithm is driven by a $\frac{1}{1-\epsilon}$ best response dynamic on
the $\Ps$-game. This algorithm is similar to that of \citet{Chien2011}.
We prove that this algorithm computes a $\frac{d!}{1-\epsilon}$-approximate
pure Nash equilibrium in a polynomial runtime parameterized by the three
constants $d,$ $A$ and $W$. This generalizes the runtime result of \citet{Chien2011} that
is only  for
symmetric congestion games with $\alpha$ bounded jump latency functions.

The second algorithm is driven by incorporating the idea of approximate potential functions,
we propose a refined best response dynamic, see Algorithm~\ref{alg:Best Response Dynamics of appro-PNE}. This algorithm defines directly on the weighted congestion game, but not on its
$\Ps$-game. We prove that this algorithm computes a $\frac{2\cdot W\cdot (d+1)}{(2\cdot W+d+1)\cdot (1-\epsilon)}$-approximate pure Nash equilibrium also in a polynomial runtime.
While
this runtime is still parameterized by the three constants $d,$ $A$ and $W$,
it is much shorter than that of Algorithm~\ref{alg:Best Response Dynamics}.
%In particular, Algorithm~\ref{alg:Best Response Dynamics of appro-PNE} computes the best known approximate pure Nash equilibria.

In fact, our Algorithm~\ref{alg:Best Response Dynamics of appro-PNE} can also compute an approximate pure Nash equilibrium with a  better approximation ratio $\rho$ in a similar polynomial runtime when the corresponding $\rho$-approximate potential function exists.
This then raises an interesting question if there is a $\rho$-approximate
potential function $\Phi(\cdot)$ for weighted congestion games fulfilling Conditions~\ref{con:restrictions on weights}--\ref{con:restrictions on latency functions}
for some constant $\rho\in(1,\frac{2\cdot W\cdot (d+1)}{2\cdot W+d+1}).$
%At present, we are unable to address this question. We thus would like to leave it for a future work.
%\subsection{Atomic congestion games}

\section*{Acknowledgement}

The first author acknowledges support from the National Natural Science Foundation of China with grant No.~12131003.
The second author acknowledges support from the National Natural Science Foundation of China with grant No.~61906062, support from the Natural Science Foundation of Anhui with grant No.~1908085QF262, and support from the Talent Foundation
of Hefei University with grant No.~1819RC29.
The third author acknowledges support from the National Natural Science Foundation of China with grants No.~12131003 and No.~11871081.
The fourth author acknowledges support from the National Natural Science Foundation of China with grant No. 72192800.

\bibliographystyle{plainnat}
\bibliography{arXiv_RWXY_Computing_Equilibria}

\newpage
\begin{appendix}
	\section{Appendices: Mathematical Proofs}
	
\subsection{Proof of Lemma~\ref{le:The relationship between the cost of O/D pairs}}
	\label{app:Proof_Of_Lemma_relationship_between_cost_of_OD_pairs}

Consider now an arbitrary iteration $t\le T_{\epsilon}(S^{(0)})$. Let $u_t$ be the player chosen in iteration $t,$
	and let $u\in\N$ be an arbitrary player.

If $u$ coincides with $u_t$, then
		 \bea
\hc_{u_t}(S^{(t)})-\hc_{u_t}(S^{(t+1)})&=&
		 \hc_{u_t}(S^{(t)})-\hc_{u_t}(s_{u_t}^{*},S_{-u_t}^{(t)})
		 \ge\epsilon\cdot \hc_{u_t}(S^{(t)})\nn\\
&=&\epsilon\cdot \hc_{u}(S^{(t)})\ge\frac{\epsilon}{W+d!\cdot W^{d+1}}
		\cdot \hc_{u}(S^{(t)}).\nn
\eea
		This follows since $u=u_t$ has a $\frac{1}{1-\epsilon}$-move of $s_{u_t}^*,$ and since
		the upper bound $W$ in Condition~\ref{con:restrictions on weights} is
		not smaller than $1.$ Hence, Lemma~\ref{le:The relationship between the cost of O/D pairs} holds
		in this special case.
		
		Now we assume that $u\ne u_t,$ and distinguish two cases below.

\textbf{Case 1: $u_t\ne u$ and $u$ has $\frac{1}{1-\epsilon}$-moves}
		
		Since $u_t$ is chosen by Algorithm~\ref{alg:Best Response Dynamics}, we obtain that
		\[
		\hat{c}_{u_t}(S^{(t)})-\hat{c}_{u_t}(S^{(t+1)})
		\ge \hc_u(S^{(t)})-\hc_u(s'_u,S_{-u}^{(t)})
		\ge \epsilon\cdot \hc_u(S^{(t)})\ge \frac{\epsilon}{W+d!\cdot W^{d+1}}\cdot \hc_u(S^{(t)}),
		\]
		where $s'_u$ is a $\frac{1}{1-\epsilon}$-move of player $u$ w.r.t. $S_{-u}^{(t)}.$
		This follows since %$S^{(t+1)}=(s_{u_t}^*,S_{-u_t}^{(t)})$ with
		$s_{u_t}^*$ is a best response of player $u_t$ w.r.t. $S_{-u_t}^{(t)},$ and since $u_t\in\N^*_t$ is a player that has a $\frac{1}{1-\epsilon}$-move with a maximum cost reduction w.r.t. these of the other players in $\N^*_t$.
		
		This completes the proof of Case 1.
		
		\textbf{Case 2: $u_t\ne u$ and $u$ does not have a $\frac{1}{1-\epsilon}$-move}
		
Note that
			\bea\label{big move}	\hc_{u_t}(S^{(t+1)})=\hc_{u_t}(s_{u_t}^*,S_{-u_t}^{(t)})<(1-\epsilon)\cdot \hc_{u_t}(S^{(t)}),
			\eea
			and that
			\bea\label{small move}
			\hc_{u}(s_{u_t}^*,S^{(t)}_{-u})\ge(1-\epsilon)\cdot \hc_u(S^{(t)}),
			\eea
			where $s_{u_t}^*$ is again a best-response of player $u_t$ w.r.t. $S_{-u_t}^{(t)}.$
			Inequalities~\eqref{big move}--\eqref{small move} follow
			since $s_{u_t}^*$ is a $\frac{1}{1-\epsilon}$-move of the selected player $u_t,$ since $u$ and $u_t$ have the same set of strategies
			(so $s_{u_t}^*$ is also a strategy of player $u$), and since
			$u$ does not have a $\frac{1}{1-\epsilon}$-move, and so moving to $s_{u_t}^*$
			can reduce the cost of player $u$ at a rate of at most $\epsilon$.

Note also that%Since $w_u\ge1$ for all $u\in\N$, we have
\bea\label{eq:the relation of the cost of edge of same strategy}
\hc_{u}(s_{u_t}^{*},S^{(t)}_{-u})&=&w_u\cdot \sum_{e\in s_{u_t}^{*}}\sum_{k=0}^da_{e,k}\cdot \Ps_k(U_e(s_{u_t}^{*},S^{(t)}_{-u}))\nn\\
&\le&
w_u\cdot \sum_{e\in s_{u_t}^{*}}\sum_{k=0}^da_{e,k}\cdot (1+k!\cdot W^k)\cdot \Ps_k(U_e(s_{u_t}^{*},S^{(t)}_{-u_t}))\nn\\
&\le&
(W+d!\cdot W^{d+1})\cdot w_{u_t}\cdot \sum_{e\in s_{u_t}^{*}}\sum_{k=0}^da_{e,k}\cdot \Ps_k(U_e(s_{u_t}^{*},S^{(t)}_{-u_t}))\nn\\
&=&(W+d!\cdot W^{d+1})\cdot\hc_{u_t}(s_{u_t}^{*},S^{(t)}_{-u_t}).
\eea
Here, we used that $w_v\in [1,W]$ for all $v\in\N$, that
 $U_e(s_{u_t}^{*},S^{(t)}_{-u})=U_e(s_{u_t}^{*},S^{(t)}_{-u_t})\cup\{W\}$ in the worst case for
 each $e\in s_{u_t}^*$, %that $x+y\le x\cdot (y+1)$ for $x,y\ge 1,$
and thus, that
\bea
\Ps_k(U_e(s_{u_t}^{*},S^{(t)}_{-u}))&\le &
\Ps_k(U_e(s_{u_t}^{*},S^{(t)}_{-u_t})\cup\{W\})\nn\\
&\le&[\Ps_k(U_e(s_{u_t}^{*},S^{(t)}_{-u_t}))^{1/k}+\Ps_k(\{W\})^{1/k}]^k\nn\\
&\le&[\Ps_k(U_e(s_{u_t}^{*},S^{(t)}_{-u_t}))^{1/k}\cdot(1+\Ps_k(\{W\}))^{1/k}]^k\nn\\
&=&\Ps_k(U_e(s_{u_t}^{*},S^{(t)}_{-u_t}))\cdot(1+\Ps_k(\{W\}))\nn\\
&=&(1+k!\cdot W^k)\cdot\Ps_k(U_e(s_{u_t}^{*},S^{(t)}_{-u_t}))\nn
\eea
for each $e\in s_{u_t}^*.$ The second inequality follows since Lemma~\ref{le:the properties of the Ph-function}c.

Inequalities~\eqref{big move}--\eqref{eq:the relation of the cost of edge of same strategy}
	together yield that
\bea
			(1-\epsilon)\cdot \hc_u(S^{(t)})&\le&
			\hc_{u}(s_{u_t}^{*},S_{-u}^{(t)})\le
			(W+d!\cdot W^{d+1})\cdot \hc_{u_t}(s_{u_t}^{*},S_{-u_t}^{(t)})\nn\\
&<& (W+d!\cdot W^{d+1})\cdot (1-\epsilon)
			\cdot \hc_{u_t}(S^{(t)}).\nn
\eea
			Hence,
			$$
			\hc_u(S^{(t)})< (W+d!\cdot W^{d+1})\cdot \hc_{u_t}(S^{(t)}).
			$$
			Moreover,
			$$
			\hc_{u_t}(S^{(t)})-\hc_{u_t}(S^{(t+1)})\ge\epsilon \cdot \hc_{u_t}(S^{(t)})>\frac{\epsilon}{W+d!\cdot W^{d+1}}\cdot \hc_u(S^{(t)}).
			$$
			This completes the proof of Case 2, and finishes the proof of Lemma~\ref{le:The relationship between the cost of O/D pairs}. $\hfill\square$

	\subsection{Proof of Lemma~\ref{le:The relationship between the cost of multi-O/D}}
	\label{proof:The relationship between the cost of multi-O/D}
Consider an arbitrary iteration $t<T_\epsilon(S^{(0)})$ for an arbitrary constant
	$\epsilon>0$ and an arbitrary initial profile $S^{(0)}.$
	Let $u\in\N$ be an arbitrary player. Similar to that of Lemma~\ref{le:The relationship between the cost of O/D pairs}, we distinguish two cases below.

\textbf{Case 1: player $u$ has  $\frac{1}{1-\epsilon}$-moves}
			
			Note that the selected player $u_t^*$ has a maximum cost reduction w.r.t. the other players
			in $\N_t^*$
			when he/she unilaterally moves to the best response strategy $s_{u_t}^*,$
			and that $u\in\N_t^*$ in this case since he/she has a $\frac{1}{1-\epsilon}$-move.
			Hence,
$$\hc_{u_t}(S^{(t)})-\hc_{u_t}(s_{u_t}^{*},S^{(t)}_{-u_t})\ge
			\hc_{u}(S^{(t)})-\hc_{u}(s_{u}^{*},S_{-u}^{(t)})>
			\epsilon\cdot \hc_{u}(S^{(t)})\ge\frac{\epsilon}{\mu}\cdot \hc_{u}(S^{(t)}).$$
			Here, note that the constant $\mu$ in Lemma~\ref{le:The relationship between the cost of multi-O/D} is not smaller than $1.$
			
			\textbf{Case 2: player $u$ does not have a $\frac{1}{1-\epsilon}$-move}

In this case, $u\notin \N_t^*.$ In particular, when
		player $u$ has the same strategy set with player $u_t,$ then
			an identical argument to that in proof of Lemma~\ref{le:The relationship between the cost of multi-O/D} applies. Hence, we assume, w.l.o.g., that
			$\Sigma_{u}\ne\Sigma_{u_t}$ for this case.

As player $u$ does not have a $\frac{1}{1-\epsilon}$-move, we obtain that
			\bea\label{small move of multi O/D}
			\hc_{u}(s'_{u},S^{(t)}_{-u})\ge (1-\epsilon)\cdot \hc_u(S^{(t)})
			\eea
			for an arbitrary strategy $s'_{u}\in\Sigma_u$.
			
			Similar to that in proof of Lemma~\ref{le:The relationship between the cost of O/D pairs},
			we now compare $\hat{c}_u(s'_{u},S_{-u}^{(t)})$ and
			$\hc_{u_t}(s_{u_t}^*,S_{-u_t}^{(t)})$ for an arbitrary strategy
			$s'_u\in\Sigma_{u}.$

Since there are totally $N$ players, and since each player has a weight smaller than $W,$
 we obtain for an arbitrary resource $e\in s'_{u}$ that
$$\Ph_k(U_e(s'_{u},S_{-u}^{(t)}))\le
	 k!\cdot L(U_e(s'_{u},S_{-u}^{(t)}))^k\le k!\cdot N^k\cdot W^k,\quad
	\forall k=0,1,\ldots,d.$$
This together with Lemma~\ref{le:the properties of the Ph-function}a yield for each
	$e\in s'_u$ and $e'\in s_{u_t}^*$ that
\begin{equation}\label{eq:the relation of the cost of edge}
	\begin{split}
		\hc_e(s'_{u},S_{-u}^{(t)})&=\sum_{k=0}^d
		a_{e,k}\cdot \Ps_{k}(U_e(s'_{u},S_{-u}^{(t)}))
		\le \sum_{k=0}^d
		a_{e,k}\cdot k!\cdot N^k\cdot W^k\le N^d\cdot W^d\cdot d!\cdot \sum_{k=0}^d a_{e,k}\\
		&\le N^d\cdot W^d\cdot (d+1)\cdot A\cdot d!\cdot \sum_{k=0}^d
		a_{e',k}\cdot \Ps_k(U_{e'}(s_{u_t}^*,S_{-u_t}^{(t)}))\\
		&=N^d\cdot W^d\cdot (d+1)!\cdot A\cdot \hc_{e'}(s_{u_t}^*,S_{-u_t}^{(t)}).
	\end{split}
\end{equation}
Here, we used that $\Ps_k(U_{e'}(s_{u_t}^*,S_{-u_t}^{(t)}))\ge 1$ for each
$k=0,1,\ldots,d,$ that
$a_{e,k'}\le A\cdot a_{e',k}$ for two arbitrary $k,k'\in \{0,1,\ldots,d\}$ with $a_{e',k}>0,$
and that there is at least one $k\in \{0,1,\ldots,d\}$ with
$a_{e',k}>0.$

Inequality~\eqref{eq:the relation of the cost of edge} implies immediately
			for an arbitrary strategy $s'_u\in \Sigma_{u}$
			 that
	\begin{equation}\label{eq:Asymmetric_Case_Important_Step}
				\begin{split}
					\hc_{u}(s'_u,S_{-u}^{(t)})\le
					E\cdot A\cdot(1+d)!\cdot N^d\cdot W^{d+1}\cdot \hc_{u_t}(s_{u_t}^*,S_{-u_t}^{(t)})
					=\mu\cdot \hc_{u_t}(s_{u_t}^*,S_{-u_t}^{(t)}).
				\end{split}
			\end{equation}
		Here, we note that $s_{u_t}^*$ may include only one resource, and
	$s'_u$ may contain $E=|\E|$ resources.
This, combined with inequality~\eqref{small move of multi O/D}, yield that
		\begin{equation*}
			\begin{split}
				(1-\epsilon)\cdot \hc_{u_t}(S^{(t)})>\hc_{u_t}(S^{(t+1)})=
				\hc_{u_t}(s_{u_t}^*,S_{-u_t}^{(t)})\ge \frac{1}{\mu}\cdot
				\hc_{u}(s'_u,S_{-u}^{(t)})\ge \frac{1-\epsilon}{\mu}\cdot
				\hc_{u}(S^{(t)})
			\end{split}
		\end{equation*}
	for an arbitrary $u\notin\N_t^*.$ Moreover, we obtain also that
	\begin{equation*}
		\hc_{u_t}(S^{(t)})-\hc_{u_t}(S^{(t+1)})>\epsilon\cdot \hc_{u_t}(S^{(t)})\ge\frac{\epsilon}{\mu}\cdot
		\hc_{u}(S^{(t)})
	\end{equation*}
when $u\notin \N_t^*.$

This completes the proof of Lemma~\ref{le:The relationship between the cost of multi-O/D}.
$\hfill\square$

\subsection{Proof of Lemma~\ref{le:The relationship between the cost of O/D pairs with appro}}
\label{app:The relationship between the cost of O/D pairs with appro}
This proof is similar to that in Appendix~\ref{app:Proof_Of_Lemma_relationship_between_cost_of_OD_pairs}. We thus only sketch the main steps below.

Let $u\in\N$ be an arbitrary player, and let $t<T_{\rho,\epsilon}(S^{(0)})$ be an arbitrary
iteration.
We assume w.l.o.g. that player $u$ does not have a
$\frac{\rho}{1-\epsilon}$-move, and so
\bea\label{eq:no appro move}
\rho\cdot  c_u(s_{u_t}^{*},S_{-u}^{(t)})\ge(1-\epsilon)\cdot c_u(S^{(t)}).
\eea
Here, we note that $s_{u_t}^*$ is also a strategy of player $u$ when
$\Sigma_{u}=\Sigma_{u_t}.$

Note that
\bea\label{eq:the relation of no appro move}
c_u(s_{u_t}^{*},S^{(t)}_{-u})&=&w_u\cdot\sum_{e\in s_{u_t}^{*}}c_e(s_{u_t}^{*},S_{-u}^{(t)})=
w_u\cdot \sum_{e\in s_{u_t}^{*}}\sum_{k=0}^da_{e,k}\cdot L(U_e(s_{u_t}^{*},S^{(t)}_{-u}))^k\nn\\
&\le& w_u\cdot \sum_{e\in s_{u_t}^{*}}\sum_{k=0}^da_{e,k}\cdot  [L(U_e(s_{u_t}^{*},S^{(t)}_{-u_t}))+W]^k\nn\\
&\le& w_u\cdot \sum_{e\in s_{u_t}^{*}}\sum_{k=0}^da_{e,k}\cdot (1+W)^k\cdot  L(U_e(s_{u_t}^{*},S^{(t)}_{-u_t}))^k\nn\\
&\le& (1+W)^{d+1}\cdot w_{u_t}\cdot \sum_{e\in s_{u_t}^{*}}\sum_{k=0}^da_{e,k}\cdot   L(U_e(s_{u_t}^{*},S^{(t)}_{-u_t}))^k\nn\\
&=&(1+W)^{d+1}\cdot c_{u_t}(s_{u_t}^{*},S_{-u_t}^{(t)}).\nn
\eea
Here, we used that $x+y\le x\cdot (y+1)$ when $x,y\ge 1$, and that
$L(U_e(s_{u_t}^*,S_{-u}^{(t)}))\le L(U_e(s_{u_t}^*,S_{-u_t}^{(t)}))+W$ when
$e\in s_{u_t}^*.$

This together with inequality ~\eqref{eq:no appro move} yield that
\begin{equation*}
(1-\epsilon)\cdot c_u(S^{(t)})\le\rho\cdot  c_u(s_{u_t}^{*},S_{-u}^{(t)})\le (1+W)^{d+1}\cdot \rho
\cdot c_{u_t}(s_{u_t}^*,S_{-u_t}^{(t)})<(1+W)^{d+1}\cdot (1-\epsilon) \cdot c_{u_t}(S^{(t)}),
\end{equation*}
which, in turn, implies that
\[
c_u(S^{(t)})<(1+W)^{d+1}\cdot c_{u_t}(S^{(t)}).
\]
Here, we used the fact that $s_{u_t}^*$ is a $\frac{\rho}{1-\epsilon}$-move of the selected player
$u_t.$
Moreover, we obtain that
\[
c_{u_t}(S^{(t)})-\rho \cdot c_{u_t}(s_{u_t}^{*},S_{-u_t}^{(t)})>\epsilon\cdot  c_{u_t}(S^{(t)})>\frac{\epsilon}{(1+W)^{d+1}}\cdot c_u(S^{(t)}).
\]

This completes the proof of Lemma~\ref{le:The relationship between the cost of O/D pairs with appro},
due to the arbitrary choice of the player $u\in \N.$
$\hfill\square$

\subsection{Proof of Lemma~\ref{le:The relationship between the cost of multi-O/D with appro}}
\label{app:The relationship between the cost of multi-O/D pairs with appro}
Let $u\in\N$ be an arbitrary player, and let $t<T_{\rho,\epsilon}(S^{(0)})$ be an arbitrary
iteration. Since this proof is similar to that in Appendix~\ref{proof:The relationship between the cost of multi-O/D}, we also sketch the main steps below.

We assume, w.l.o.g., that
			$\Sigma_{u}\ne\Sigma_{u_t}.$ As player $u$ does not have a $\frac{\rho}{1-\epsilon}$-move, we obtain that
			\bea\label{small move of multi O/D with appro}
			\rho\cdot c_{u}(s'_{u},S^{(t)}_{-u})\ge (1-\epsilon)\cdot c_u(S^{(t)})
			\eea
			for an arbitrary strategy $s'_{u}\in\Sigma_u$.

For each $s'_u\in\Sigma_{u},$
\bea\label{eq:the relation of the cost of edge with appro}
&&c_u(s'_{u},S^{(t)}_{-u})=w_u\cdot\sum_{e\in s'_{u}}c_e(s'_{u},S_{-u}^{(t)})=
w_u\cdot \sum_{e\in s'_{u}}\sum_{k=0}^da_{e,k}\cdot L(U_e(s'_{u},S^{(t)}_{-u}))^k\nn\\
&\le& w_u\cdot N^d\cdot W^d\cdot \sum_{e\in s'_u}\sum_{k=0}^d a_{e,k}
\le E\cdot A\cdot (1+d) \cdot N^d\cdot W^{d+1}\cdot w_{u_t} \cdot \sum_{e'\in s^*_{u_t}}
\sum_{k'=1}^d a_{e',k'}\nn\\
&\le& E\cdot A\cdot (1+d)\cdot N^d\cdot W^{d+1}\cdot c_{u_t}(s_{u_t}^*,S_{-u_t}^{(t)})=\varpi\cdot c_{u_t}(s_{u_t}^{*},S_{-u_t}^{(t)}),\nn
\eea
which together with inequality ~\eqref{eq:no appro move} yield that
\begin{equation*}
(1-\epsilon)\cdot c_u(S^{(t)})\le\rho\cdot  c_u(s'_{u},S_{-u}^{(t)})\le \varpi\cdot \rho
\cdot c_{u_t}(s_{u_t}^*,S_{-u_t}^{(t)})<\varpi\cdot (1-\epsilon) \cdot c_{u_t}(S^{(t)}),
\end{equation*}
which, in turn, implies that
\[
c_u(S^{(t)})<\varpi\cdot c_{u_t}(S^{(t)}).
\]
Moreover, we obtain that
\[
c_{u_t}(S^{(t)})-\rho\cdot c_{u_t}(s_{u_t}^{*},S_{-u_t}^{(t)})>\epsilon\cdot  c_{u_t}(S^{(t)})>\frac{\epsilon}{\varpi}\cdot c_u(S^{(t)}).
\]

This completes the proof of Lemma~\ref{le:The relationship between the cost of multi-O/D with appro},
due to the arbitrary choice of player $u$ from $\N.$
$\hfill\square$
\end{appendix}
\end{document}